\documentclass[journal]{IEEEtran}

\usepackage{amssymb,amsthm,amsfonts,amsbsy}
\usepackage{geometry}
\usepackage{hyperref}
\usepackage{graphicx}
\usepackage{mathrsfs}
\usepackage{booktabs}
\usepackage{enumitem}
\usepackage{mathtools}  

\DeclareMathAlphabet\mathbfcal{OMS}{cmsy}{b}{n}
\usepackage{framed}
\usepackage{enumitem}
\usepackage{tocloft}
\usepackage{tikz}

\usepackage{float}
\usepackage{comment}
\usepackage{tikz}

\usepackage{lipsum}

\title{\textbf{The Gaussian Multiuser Channel Capacity Region}}

\author{%
John~M.~Cioffi\thanks{J.~M.~Cioffi is with Stanford University.
This work was financially supported by funding from Samsung, Intel, and Ericsson.}
-- invited paper, {\em IEEE Transactions on Communication} \vspace{-.4in}}
\date{November 3, 2026}

\newcommand{\xvec}{\mbox{\boldmath ${x}$}}

\newcommand{\yvec}{\mbox{\boldmath ${y}$}}

\newcommand{\nvec}{\mbox{\boldmath ${n}$}}
\newcommand{\vvec}{\mbox{\boldmath ${v}$}}

\newcommand{\zvec}{\mbox{\boldmath ${z}$}}
\newcommand{\bvec}{\mbox{\boldmath ${b}$}}

\newcommand{\thetavec}{\mbox{\boldmath ${\theta}$}}

\newcommand{\rxx}{R_{\xvec \xvec}}

\newcommand{\beq}{\begin{equation}}
\newcommand{\eeq}{\end{equation}}

\newcommand{\isdef}{\stackrel{\Delta}{=}}

\newcommand{\cI}{{\cal I}}

\newtheorem{theorem}{Theorem}

\begin{document}

\maketitle
\begin{abstract}
This work characterizes a Gaussian multiuser-channel capacity-region using receiver-subset message atoms, Gaussian single-user equivalents, and successive decoding. A memoryless multiuser-channel model specifies both the channel conditional distribution and the independent-use input distribution.  For all linear Gaussian channels, independent Gaussian user signaling and receiver-local worst-case-noise/MMSE-GDFE transformations provide a single-user-equivalent basis for the complete capacity-region construction. Each user distributes its input autocorrelation among a finite set of receiver-subset atoms, whose rate sums satisfy the applicable receiver constraints. Minimum-energy and weighted-rate-sum designs optimize these allocations, while fixed-allocation linear programs determine feasible atomic rates and synchronized time-shared successive-decoding implementations. The resulting framework includes Gaussian multiple-access, broadcast, and interference channels, with numerical examples illustrating their rate regions and receiver designs.
\end{abstract}
\vspace{-.2in}

\section{Introduction\label{sec1}}

The multiuser channel has a general capacity region with a compact
structural description.  Traditional formulations often introduce
auxiliary random variables whose operational meaning and cardinalities
can obscure the underlying geometry.  This paper instead assigns rates
to receiver-subset atoms and derives the converse from single-user
bounds and chain-rule expansions
\cite{fano1961transmission,coverbook,csiszar2011information}.  The
framework applies to memoryless multiuser channels without feedback.

As in synchronized block-coded communication systems
\cite{tse2005fundamentals}, all participating transmitters and
receivers use a common codeword/block timing.  Consequently both
converse and achievability use one global time-sharing variable that
synchronizes receiver decoding-order realizations; asynchronous
decoding is a different model and not addressed.

This paper defines a memoryless multiuser channel (MMUC) by specifying both its physical channel and its admissible independent-use signaling ensemble. Reliable coding may extend over arbitrarily many channel symbols, but such coding does not change the prescribed elementary-symbol statistics. For Gaussian channels, the elementary symbol may itself be a vector containing multiple antennas or frequency subcarriers (or ``tones''). This distinction is important when comparing multiuser-coding methods, because temporal dependence in one user's transmitted signal can change the interference statistics experienced by another user.

The $U$-receiver receiver-subset atomic representation has at most $U\cdot 2^{U-1}$ atoms. The Gaussian design determines their autocorrelation matrices and rates subject to the receiver decoding requirements. A priority vector $\boldsymbol\theta$ identifies supporting rate-region boundaries, while receiver breakpoints and synchronized successive MMSE-GDFE decoding provide the corresponding canonical implementation structure. This paper's chain-rule/GDFE viewpoint builds on \cite{cioffi1994gdfe,cioffi1997gdfe,cioffi_ee379}.
\vspace{-.3in}
\subsection{Main Contributions\label{sec1.1}}
This work has five primary contributions: 
\begin{enumerate}
\item \textbf{Receiver-subset atoms:}
A finite atomic representation assigns each user's information to the receivers required to decode it. The complete representation has at most $U\cdot 2^{U-1}$ atoms, independent of the particular Gaussian channel.
\item \textbf{Gaussian user optimality:}
Block-diagonal input autocorrelation matrix and/or
worst-case-noise and receiver-local MMSE-GDFE transformations find Gaussian multiuser codes optimum via information-lossless transformations to single-user equivalents, while preserving the applicable transmitter-independence and receiver-separation constraints.
\item \textbf{Constructive atomic rate/energy design:}
Canonical design methods appear for both minimum-energy and weighted-rate-sum optimizations. At fixed atomic autocorrelation matrices, the atomic-rate problem can reduce to a linear program.
\item \textbf{Receiver breakpoints and successive decoding:}
One user-priority ordering identifies receiver breakpoints. A second linear program tests synchronized successive MMSE-GDFE realizations at fixed atomic autocorrelation matrices, permitting changes in atomic rates without changing the transmitted Gaussian user-level signaling statistics, nor losing optimality.
\item \textbf{Constructive numerical examples: }
Two- and three-user Gaussian interference channels and frequency-selective channels illustrate the capacity-region calculations, atomic allocations, and receiver-specific decoding schedules.
\end{enumerate}
\vspace{-.15in}
\subsection{Reader Guide and Converse Scope\label{sec1.2}}
Section~\ref{sec2} specifies the memoryless multiuser-channel model and receiver-subset atom representation. Section~\ref{sec3} develops the Gaussian single-user equivalents and Gaussian-input optimality. Section~\ref{sec4} gives the atomic rate--energy design, receiver breakpoints, successive-decoding implementations, and representative Gaussian channel examples. Section~\ref{sec5} compares the construction with earlier Gaussian interference-channel methods. The receiver-subset structure is established independent of a prescribed decoding architecture; successive MMSE-GDFE provides a canonical Gaussian implementation for admissible rate and/or energy constraints.
\vspace{-.15in}
\subsection{Notation, Structural Objects, and a 3-User Example\label{sec1.3}}

Users have indices $u \in[1:U]\triangleq\{1,\ldots,U\}$.  Transmitter $u$ owns atom
$a=(u,S)$ where receivers in subset $S$, which contains $u$, decode it. 
The complete atom set is
\beq
\mathcal A \triangleq \{(u,S): u \in S \subseteq [1:U]\},
\label{defineAset}
\eeq
with $| \mathcal A|=U\cdot 2^{U-1}$.  Atom $a=(u,S)$ carries rate
$b_a\equiv b_{u,S}$, and user rates are the ``projection''\footnote{The notation $S \ni i$ sums over all $S$ that contain $i$.} (rate sum)
\beq
b_i=\sum_{S\ni i} b_{i,S} 
\label{batomsum}
\eeq
from (up to) $2^{U-1}$-dimensional atomic rate region to a $U$-dimensional achievable user rate region.
Receiver $r$ decodes the set of atoms
\[
\mathcal A_r\triangleq\{a\in\mathcal A:r\in S(a)\} \; .
\label{atomreq}
\]
The induced marginal channel averages over undecoded atoms.

For $U=3$, Fig.~\ref{periodic} enumerates the 12 possible atoms:
three global atoms decoded by all receivers, six pairwise atoms, and
three private atoms.  This ``periodic atom table'' is channel-independent;
the channel, target rates, and supporting priority determine which
atoms activate.  The example on the right activates only four of the
12 possible atoms.

\begin{figure*}[t]
\begin{framed}
\centering
\includegraphics[width=0.9\linewidth]{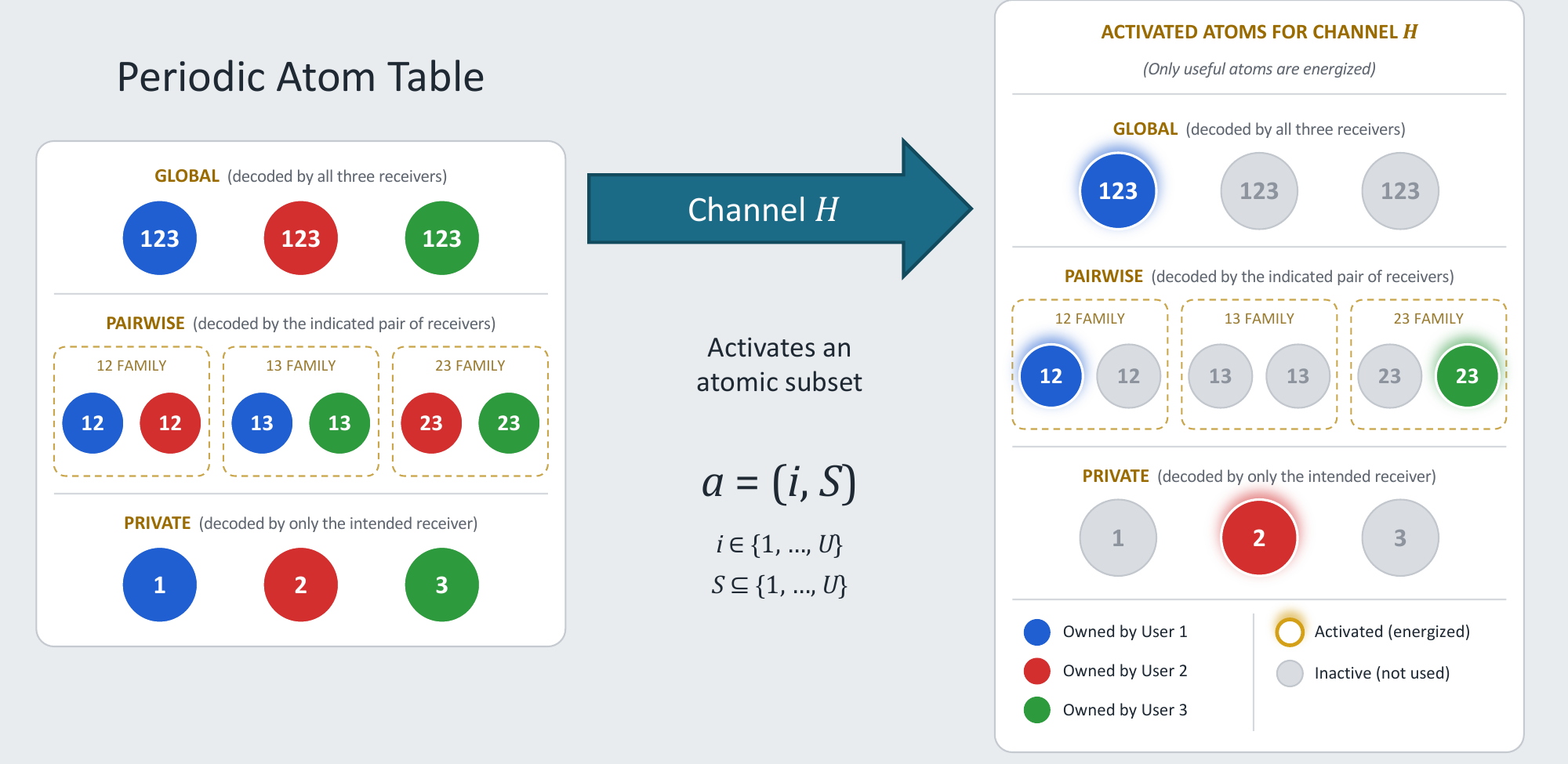}
\caption{Multiuser periodic atom table and one channel realization.
The channel does not create information atoms; it selects and
activates them to meet user-rate requirements.}
\vspace{-.1in}
\label{periodic}
\end{framed}
\vspace{-.3in}
\end{figure*}
\vspace{-.1in}

\begin{figure*}[t]
\begin{framed}
\centering
\includegraphics[width=0.82\linewidth]{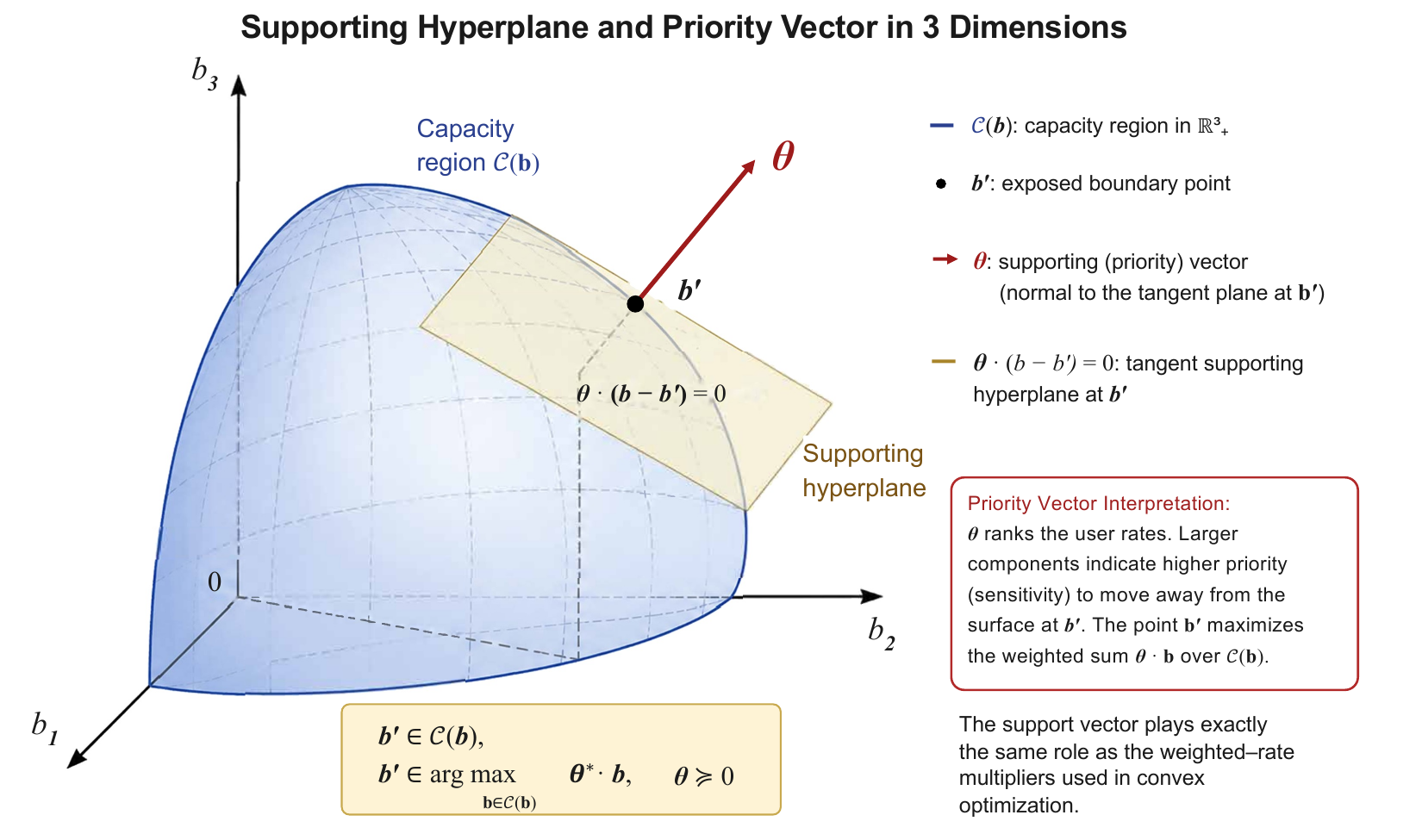}
\caption{Supporting hyperplane and priority vector for
a three-user rate region. The nonnegative vector
$\boldsymbol{\theta}$ specifies the weighted user-rate sum,
and the supported boundary point $\mathbf{b}'$ maximizes
$\boldsymbol{\theta}^{t}\mathbf{b}$.
A supporting plane may touch the region at one point
or along a face. The same priority vector guides
the Gaussian atomic rate--energy design and receiver
breakpoint ordering.}
\label{fig:tangent}
\end{framed}
\end{figure*}

\section{Memoryless Multiuser Channels and Receiver-Subset Atoms\label{sec2}}

A memoryless multiuser channel (MMUC) differs from a single-user channel in an
important respect: each transmitter's signal is both desired information
at its intended receiver and, potentially, interference at other receivers.
Consequently, the temporal statistics of one user's signaling affect the
other users' statistical channel characterizations.  This section specifies a
MMUC model that
distinguishes channel symbols from reliable-code block lengths.  This section also elaborates upon Section \ref{sec1}'s
finite receiver-subset atom representation.
\vspace{-.08in}
\subsection{A Memoryless Multiuser-Channel Model\label{sec2.1}}

A conventional memoryless-channel condition \cite{elgamalbook} factors the
physical transition distribution,
\begin{equation}
 p(\yvec^{(K)}\mid \xvec^{(K)})
   =\prod_{k=1}^{K}p(\yvec_k \mid  \xvec_k),
 \label{eq:physical-memoryless}
\end{equation}
where $\xvec_k=(\xvec_{1,k},\ldots, \xvec_{U,k})$ and
$\yvec_k=(\yvec_{1,k},\ldots,\yvec_{U,k})$.
This model extends the single-user channel model directly; however, the sequel questions
this model and provides a more appropriate alternative model. 
Specifically, Equation~(\ref{eq:physical-memoryless}) permits arbitrary temporal
statistical dependence among the transmitted symbols.  Such dependence
can create temporal dependence in a receiver's observed interference, even when the conditional distribution factors
over $k$.

For example, a Gaussian ``Z-interference'' channel has
\begin{equation}
 y_{1,k}=x_{1,k}+ h \cdot x_{2,k}+n_{1,k} \; \mbox{and} \;\;  y_{2,k}=x_{2,k} \; ,
 \label{eq:z-memory-example}
\end{equation}
with temporally white $n_{1,k}$, receiver 1 encounters the effective
disturbance $z_{1,k}=h \cdot x_{2,k}+n_{1,k}$.  A temporally correlated
user-2 input produces a correspondingly correlated disturbance at
receiver 1.  Its temporal autocorrelation, and thus its effective receiver
channel statistics, can differ from those that correspond to independent-use
interference with the same average energy.  (\ref{eq:physical-memoryless})'s distribution nevertheless remains unchanged.

Receiver 1 may significantly increase user 1's rate by ``whitening'' the interference $h \cdot x_{2,k} + n_{1,k}$, with small user 2
rate reduction.  For this to happen, user 2 must correlate its ``independent'' or memoryless transmissions over $k$.  
Good design recognizes this interference should be first whitened to create a white-noise-equivalent that becomes ``memoryless'' over the
now larger symbol length.  Otherwise, colored noise is simply a different channel, which has a well known higher capacity than
a white-noise channel with the same noise energy.  Conflating these two different channels and claiming a ``multi-letter'' coding
gain is inconsistent with the use of the term ``memoryless.''   Instead, long-known good design simply increases the symbol length,
appropriately using guard bands, cyclic prefixes, tail-biting, or generally packet-separation methods and consistently sets input and channel
symbol lengths equal, thereby averting artificially creating a conundrum.

More generally, any finite $M$-letter construction for a
memoryless linear Gaussian MUC can be represented as a single
vector-symbol use of its $M$-fold Cartesian-product channel.
For a channel matrix $H$, the repeated channel has
\begin{equation}
 H^{(M)}=I_M\otimes H
   =\operatorname{diag}(\underbrace{H,\ldots,H}_{M}),
 \label{eq:cartesian-repeat}
\end{equation}
where each user's $M$ successive inputs aggregate
into one enlarged vector symbol.  Any correlation among
these $M$ inputs has then an enlarged
input autocorrelation matrix.  The Gaussian vector-channel
capacity-region formulation developed in this paper
applies directly to this Cartesian-product channel.
Consequently, the present formulation encompasses finite multi-letter constructions and does
not omit them through a coding-blocklength restriction.

A more realistic and physically appealing MMUC model instead fixes the independent-use
signaling model together with the conditional distribution or more formally:\\
\textbf{Definition (Memoryless Multiuser Channel, MMUC):}
An MMUC establishes an elementary vector-symbol interval and an
admissible independent-user signaling ensemble for which, conditioned on
any prescribed synchronized operating (``time-sharing'') schedule $Q$, the joint distribution
factors across elementary channel uses:
\begin{eqnarray}
 p(\mathbf x^{(K)}\,\yvec^{(K)} \mid Q)
 & = &\prod_{k=1}^{K}
 \left[\prod_{u=1}^{U}p_{u,k}(\xvec_{u,k} \mid Q)\right] \nonumber \\
 && \;\;\;\;\;  \cdot p(\yvec_k \mid \xvec_k ,Q).
 \label{eq:mmuc-factorization}
\end{eqnarray}
The physical transition distribution and the admissible signaling ensemble thus
share the same elementary symbol interval, as (\ref{eq:cartesian-repeat}) alternatively permits.  The schedule $Q$ is
specified independently of the users' messages and may describe
synchronized time sharing among independent-use operating points.
The input distribution remains stationary, but its data rate (in a good random-coding sense) may vary. 
Encoding, and especially decoding, design aggregates all states asymptotically.

Equation (\ref{eq:mmuc-factorization});s MMUC does not require
independence among components of an individual vector symbol or among
receiver outputs within one symbol interval.  It requires independence
across elementary intervals, conditioned on the prescribed schedule,
and independence among the users' signaling ensembles before the
physical channel combines their inputs.  In particular, receiver noise
may be correlated within a vector elementary symbol, which clearly resolves
Equation (\ref{eq:physical-memoryless})'s traditional-model multiuser-definition conundrum. 

The elementary vector-symbol dimension is not a restriction
to a single scalar use of the original physical channel.
In particular, the Cartesian-product construction in
(\ref{eq:cartesian-repeat}) incorporates any finite number
of original channel uses into one enlarged elementary
symbol.  The MMUC factorization then applies across
successive enlarged symbols, while permitting admissible
correlation within each such symbol.

This work's MMUC therefore defines an exhaustive class of memoryless multiuser signaling models.  
Strategies whose prescribed input ensembles introduce
new temporal dependence across an infinite number of symbols are outside the
MMUC model considered here.  

{\bf Model elaboration:} Correlation of interference, and as well of noises, should initially affect the
designer's symbol-length definition.  For instance, receiver noise-whitening
then adopts a sufficiently long elementary symbol length to allow filtering
that removes memory between successive channel-symbol uses. 

Specifically, (\ref{eq:cartesian-repeat})'s Cartesian-product representation in
(\ref{eq:cartesian-repeat}) incorporates finite multi-letter
constructions within an enlarged elementary vector symbol,
rather than requiring a separate (and questionable) channel theory.  Any
possible rate improvements are consequently within
the Gaussian capacity-region formulation for the enlarged
channel.  Channels with {\bf genuine} infinite physical memory, including
Markov channels whose states persist across successive
symbols, present a different problem that this work does not consider.

\subsubsection{Elementary symbols, physical memory, and reliability\label{sec2.1.1}}

The elementary MMUC symbol may contain multiple antennas, multiple frequency
tones, or other jointly processed dimensions.  Boldface denotes these
vector elementary symbols; and their dimensionalities are left implicit unless needed.
For a frequency-selective physical channel, guard intervals, cyclic
prefixes, block processing, or noise-whitening filters may help construct
independent elementary vector-symbol intervals.  Such processing must, of course,
actually establish the required symbol independence; merely choosing a finite
block length does not eliminate arbitrary channel memory from this work's capacity region and consequent canonical designs. 

With established elementary-symbol interval, a reliable code may
span $K$ such intervals, with $K \rightarrow\infty$.  The codeword
length $K$ is distinct from any FFT size $N$, tone index $n$, or
antenna dimension.  A particular deterministic codeword is a symbol sequence, and its symbols need not appear statistically independent
when analysis averages messages over a fixed codebook.  In contrast,
independent random-codebook generation from the stipulated per-use
input distributions yields (\ref{eq:mmuc-factorization}) upon averaging
over both codebook generation and messages.  This ensemble identity is
not an assertion of exact product statistics for every selected finite
codebook.  Reliable coding across $K$ uses remains part of the model and
does not introduce an additional admissible ``multi-letter'' signaling distribution over multiple symbols.

\vspace{-.08in}
\subsection{Independent Messages and Receiver-Subset Atoms\label{sec2.3}}

Section \ref{sec1}'s Equation (\ref{defineAset})'s receiver-subset atoms define a rate-splitting
architecture,
where atom $a =(u, S)$ carries a portion of user $u$'s message that
every receiver in $S$ reliably decodes.  There are at most $U\cdot 2^{U-1}$ possible receiver-subset labels.  This
finite count describes the available atom architecture; no claim that
an arbitrary general-channel code admits a lossless independent-message
atomization is necessary here.  For the Gaussian MMUC, Section~\ref{sec3} specifies atom choices and
input autocorrelations that align with receiver decoding rules.

A {\bf multicast atom} that multiple receivers decode must satisfy every applicable
receiver constraint.  Receiver $r$'s decoding order is  $\pi_r$, and
$\Delta_r(a;\pi_r,p)$ is its conditional mutual-information
increment after earlier atoms' decoding in that order.  For a
fixed independent-use signaling distribution and synchronized schedule $Q$,
the corresponding multicast limiting mutual information is
\begin{equation}
 \mathcal I_{\min}^{\rm av}(a)
 =\min_{r\in S(a)}\mathbb E_Q
 \big[\Delta_r(a;\pi_r(Q),p_Q)\big].
 \label{eq:mmuc-Imin}
\end{equation}
The receiver-specific increments are averaged over the common
schedule \emph{before} taking the minimum across receivers; in
general the minimum and expectation do not commute.  Revealing an
already decoded atom in a mutual-information chain rule means conditioning on its known
message or codeword, not imposing a particular physical decoder.
Equation~(\ref{eq:mmuc-Imin}) describes a decoding constraint for a
specified atom construction; by itself it is not a general MMUC
capacity-region converse.

Figure \ref{fig:tangent} illustrates the priority vector $\thetavec \succeq 0$, which also has interpretation as the normal vector for a tangent hyperplane to the capacity region ${\cal C} (\bvec )$. 
While $\thetavec^* \cdot \bvec$ often represents what mathematicians call a ``support function'' that is maximum over all input choices for any and all $\thetavec$ in describing the capacity region,
it also relates the local incremental change of a user's rate relative to all others; $\thetavec$ is the gradient at the tangent vector $\bvec^\prime$.  This latter interpretation is a helpful canonical design aid in that it means the user $u^o$ with largest $\theta_{u^o} \ge \theta_{i \ne u^o}$ makes the most difference in a local rate change.  Indeed, Sections \ref{sec3} and \ref{sec4} show $\thetavec$ also is a Lagrange constraint vector for data rates when minimizing energy sum.  User $u^o$ needs most assistance from decoding order, which means it decodes last.  Similarly the user with minimum value decodes first.  This simple rule follows from single-user equivalents that arise in Section \ref{sec3} and then Section \ref{sec4} confirms it with several examples.

\section{Gaussian Channels and Single-User Equivalents\label{sec3}}

This section expands known \cite{coverbook} Gaussian single-user extremality to the
MMUC's  through lossless receiver-local and/or receiver-global transformations.
The latter transformations do not initially depend on
receiver-subset atoms, decoding orders, or a priority vector.  Only
after the relevant single-user equivalents' establishment does this work's canonical design optimize the multiuser rate vector and its atomic realization.

The Gaussian specialization has elementary vector-symbol distribution for receiver $r$
\begin{equation}
 \mathbf y_r=\sum_{u=1}^{U} H_{ru} \cdot \mathbf x_u
 +\mathbf n_r,\qquad r=1,\ldots,U,
 \label{eq:mmuc-gaussian}
\end{equation}
where the vector notation accommodates scalar, MIMO, multitone, and
combined MIMO-multitone models.  Separate user transmit locations imply the
block-diagonal input autocorrelation
\begin{equation}
 R_{xx}=\operatorname{blockdiag}
 \big(R_{x_1x_1},\ldots,R_{x_Ux_U}\big).
 \label{eq:mmuc-blockdiag}
\end{equation}
Correlation within each user's vector symbol dimensions remain admissible, but not between them.

A symbol may span $N$ dimensions, which can be frequency, within-symbol time, space or combinations thereof. In some cases like intersymbol interference channels, it may be necessary for $N \rightarrow \infty$ for certain guard-period/cyclic-prefix losses to vanish asymptotically, but it is still one symbol.  Random coding arguments pick symbols (and ultimately this may mean picking their $N$-dimensional symbols) independently for $K$ successive symbols.  That ``outer'' limit may also have $K \rightarrow \infty$ in a traditional single-user random-coding sense.  The inner limit may be come large before the outer limit becomes large.   This work at times presumes their interchange to decode one user's or atom's code with $K$ of its symbols as decoder input.  That user may not occupy all of the symbol's $N$ dimensions.   To avoid cumbersome notation by carrying indices $n$ for intrasymbol and $k$ for intersymbol, especially since the $k$ corresponds to independent successive symbols drawn from the same distribution, this work will not carry these indices except where absolutely necessary.  The reader can infer from context which limit.  As in Section \ref{sec2}, this work does not consider channel inputs correlated from symbol to symbol because any channel memory has been removed with guard bands within a symbol, and/or possibly via (\ref{eq:cartesian-repeat}).

\subsection{Single-User Gaussian Optimality\label{sec3.1}}
The vector Gaussian channel is
\begin{equation}
 \yvec=H\cdot\xvec+\nvec,\qquad
 R_{xx}=\mathbb E[\xvec\cdot\xvec^*],\qquad
 R_{nn}=\mathbb E[\nvec\cdot\nvec^*].
 \label{eq:su-model}
\end{equation}
Here the superscript $*$ denotes conjugate transpose.  The input and noise are
independent; the noise has zero mean.  Thus $R_{xx}$ is its autocorrelation matrix. (\ref{eq:mmuc-factorization})'s memoryless multichannel definition consistently continues to hold. 
\begin{figure*}[t]
\begin{framed}
\centering
\includegraphics[width=0.9\linewidth]{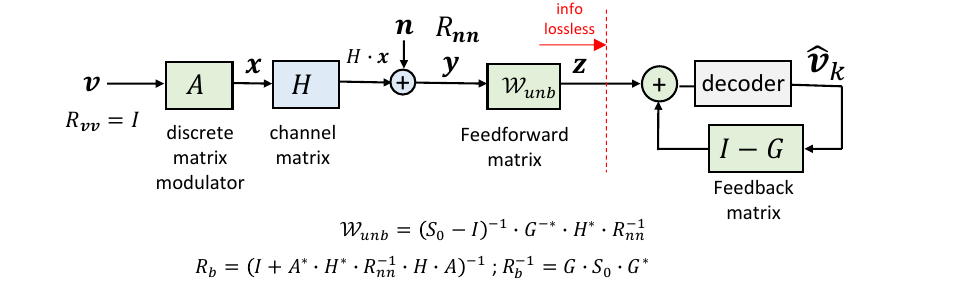}
\caption{Basic GDFE and lossless processing point. {\em There can be an individual GDFE for each $\yvec_r = H_r \cdot \xvec + \nvec_r$ also.}}
\label{fig:gdfeintro}
\end{framed}
\vspace{-.1in}
\end{figure*}
For any prescribed $R_{xx}\succeq0$ and $R_{nn}\succ0$,
maximum single-user mutual information satisfies
\begin{equation}
\begin{aligned}
\cI(\xvec;\yvec)
&\le {\cal C}(R_{xx},R_{nn})\\
=\alpha \cdot & \log_2\Bigl|
 I+R_{nn}^{-1/2} \cdot H
\cdot R_{xx}\cdot H^* \cdot R_{nn}^{-*/2}
\Bigr|.
\end{aligned}
\label{eq:su-gaussian}
\end{equation}
where $\alpha=1/2$ for real symbols and $\alpha=1$ for complex symbols.   The input is well-known \cite{coverbook} to be 
the best as a single-user Gaussian distribution for any {\bf AND ALL} $R_{xx}$.

On the rank-$d$ energized input $R_{xx}$, a design may choose a rank-$d$ square-root matrix
$A$ without loss to map the channel input $\xvec$ to a unit-variance 
information sequence also of same rank $\vvec$ such that :
\begin{equation}
 \xvec=A \cdot\vvec,\qquad A \cdot A^*=R_{xx},\qquad
 \mathbb E[\vvec\cdot\vvec^*]=I_d.
 \label{eq:su-root}
\end{equation}

\subsection{Information-Lossless Receiver Transformations\label{sec3.lossless}}
Noise whitening and/or Figure \ref{fig:gdfeintro}'s
MMSE-GDFE feedforward linear processing (to $\zvec$) all maintain the mutual information $\cI (\vvec ; \yvec) = \cI (\zvec ; \vvec )$ \cite{cioffi1997gdfe,cioffi_ee379}.
The transformations act on noisy observations; they do not provide the
receiver noiseless access to $\vvec$.  

This single-user result, and the corresponding Gaussian optimality, eventually expands to all MMUCs in the following successive subsections.
In preparation, Figure \ref{fig:gdfeintro}'s caption notes that the GDFE construct also can individually apply to receiver $r$'s input $\yvec_r$
while it decodes any, some, or all users.  In certain situations that this work exploits, the individual GDFEs become the exact components
of the larger GDFE for the entire set of channel outputs.   For any disjoint user sets
$\mathcal D$ and $\mathcal K$, the lossless processing of
mutual information yields
\begin{equation}
 \cI(M_{\mathcal D};\zvec_r\mid M_{\mathcal K})
 =\cI(M_{\mathcal D};\yvec_r\mid M_{\mathcal K}).
 \label{eq:local-lossless}
\end{equation}
The common WCN-equivalent priority order determines
each receiver's decoding order and breakpoint, while
its receiver-local MMSE-GDFE provides the corresponding
sufficient statistics using only that receiver's
input.
This equality holds irrespective of the transmitted messages' probability distribution and independently of which additional
messages receiver $r$ ultimately decodes, as long as its processing is invertible or at least a sufficient statistic.  Any dimension-reducing
transform requires only the corresponding sufficient-statistic property.
For a specified receiver $r$, a {\em receiver breakpoint} distinguishes the decoded users from those treated as interference. The breakpoint's priority-dependent form follows from the single-user Gaussian ordering in (\ref{eq:thetavecorder}). The breakpoint concerns the receiver's decoding obligations and does not alter the physical channel nor the information-lossless character of (\ref{eq:local-lossless})).

A single-user MMSE-GDFE supplies an information-preserving Gaussian
representation through its feedforward and operations.
This also holds true for a local MMSE-GDFE that only acts on $\yvec_r$, preserving (\ref{eq:local-lossless}) for
any of the $2^U$ posible ${\cal D}$ and ${\cal K}$ partition choices that a subsequent decoder might try to decode.
Returning to the entire channel output $\yvec$, a single receiver filters $\yvec$.  
This ${\cal W}$ transformation itself in general is not block diagonal.
However, when $R_{nn}$ minimizes the mutual information $\cI (\xvec ; \yvec )$,
and retains in do doing, the local noise autcorrelation values along its (block) diagonal,
the filter best ${\cal W}$ setting for this worst-case noise autocorrelation becomes
block diagonal without loss, \cite{yu2004sumcapacity,cioffi_ee379,cioffi2023singular}, as Subsection \ref{sec3.4} further details. 
The local MMSE-GDFEs then are canonical for all the partitions as they become equivalent to the larger WCN-based MMSE-GDFE.
Then, the 
receiver-local operations can be stacked as
\begin{equation}
 \mathcal W_{wcn}=\operatorname{blockdiag}
 (\mathcal W_1,\ldots,\mathcal W_U).
 \label{eq:stacked-local}
\end{equation}
This stack is invertible and therefore also preserves information in the worst-case-noise-autocorrelation case.
The block-diagonal form does not therefore exchange observations among uncoordinated
receivers.  The WCN construction below identifies when this
receiver-block processing combines into a canonical stacked
single-user MMSE-GDFE representation.

\vspace{-.2in}
\subsection{Gaussian Multiple-Access Channels\label{sec3.3}}
\vspace{-.05in}
A Gaussian multiple-access channel (MAC) has independent separately located
transmitters and a common receiver $\yvec \equiv \yvec_r$ that observes their superimposed filtered channel output,
\begin{equation}
 \yvec=\sum_{u=1}^{U}H_u\cdot\xvec_u+\nvec
       =H\cdot\xvec+\nvec .
 \label{eq:mac-model}
\end{equation}
The MAC's transmitter independence restriction imposes the block-diagonal
input autocorrelation of~\eqref{eq:mmuc-blockdiag}; it does not
require diagonal autocorrelation within any individual user's vector signal.
For every prescribed collection $\{R_{x_ux_u}\succeq0\}$,
Subsection~\ref{sec3.1}'s single-user Gaussian extremality
applies because this is equivalent to a single-user channel with a positive semidefinite
input autocorrelation matrix.
Again, this holds {\bf for ANY} $Rxx \succeq 0$, including thus block diagonal.  In particular, for every user subset
$\mathcal T\subseteq[1:U]$, independent Gaussian inputs give the
familiar MAC subset bound \cite{elgamalbook}

\begin{equation}
\begin{aligned}
\sum_{u\in\mathcal T}b_u
&\le \cI(\xvec_{\mathcal T};
 \yvec\mid\xvec_{\mathcal T^c})\\
=\alpha \cdot& \log_2\Bigl|
 I+R_{nn}^{-1/2} \cdot H_{\mathcal T}
\cdot R_{\mathcal T}
 H_{\mathcal T}^* \cdot R_{nn}^{-*/2}
\Bigr|.
\end{aligned}
\label{eq:mac-subset}
\end{equation}
where $R_{\mathcal T}$ is the indicated users' block-diagonal autocorrelation matrix.
Equation-(\ref{eq:su-gaussian})-like inequalities, {\em for all subsets}, bound
the Gaussian MAC polymatroid for this fixed $R_{xx}$.
The common receiver can implement (\ref{eq:mac-subset})'s corner points with an
MMSE-GDFE that uses successive decoding in the corresponding order;
time/dimension sharing ($Q$) obtains the intervening rate points.  Thus the
block-diagonal restriction on a single-user Gaussian channel
yields the Gaussian MAC achievable region, followed, if desired,
by optimization over admissible user autocorrelations to generate the capacity region.  The MAC needs no atomic subdivision.
The MAC users are all optimally Gaussian, each and jointly, in the random-code single-user sense.

The MAC weighted rate sum is 
\begin{equation}
\begin{aligned}
h(\boldsymbol{\theta})
&=\sum_{u=1}^{U}\theta_u \cdot b_u\\
&=\alpha \cdot \sum_{u=1}^{U}\theta_u \cdot 
 \Biggl\{
 \log_2\Bigl(\sum_{i=u}^{U}
 H_i \cdot R_{xx}(i) \cdot H_i^*\Bigr)\\
&\qquad-
 \log_2\Bigl(\sum_{i=u+1}^{U}
 H_i \cdot R_{xx}(i) \cdot H_i^*\Bigr)
 \Biggr\}\\
&=\alpha \cdot \sum_{u=1}^{U}
 (\theta_u - \theta_{u-1}) \cdot \\
&\qquad\cdot
 \log_2\Bigl(\sum_{i=u}^{U}
 H_i \cdot R_{xx}(i) \cdot H_i^*\Bigr).
\end{aligned}
\label{eq:mac-weighted}
\end{equation}
with $\theta_{-1} = 0$.
Maximization of this weighted rate sum has an ensured global optimum when the users are reordered so that
\beq
\theta_U \ge \theta_{U-1} \ge \hdots \ge \theta_1 \; .
\label{eq:thetavecorder}
\eeq
This order causes all the terms to be convex in contribution to (\ref{eq:mac-weighted})'s overall convex sum \cite{tse2005fundamentals,cioffi_ee379}.
These results are known, but help prepare in current form for new results later.

{\bf Receiver breakpoint rule:} With users indexed so that $\theta_U\geq\cdots\geq\theta_1$, receiver $r$ has its breakpoint at user $r$. It decodes its own messages and the designated multicast atoms of lower-priority users $1,\ldots,r-1$, while treating higher-priority users $r+1,\ldots,U$ as interference. Within each receiver's decoded set, successive decoding proceeds from smaller to larger applicable priority, so that the highest-priority decoded user decodes last. Equal priorities permit alternative orders. The receiver-subset atom labels and the optimized atomic allocations determine which of these possible decoding operations are actually present.

\subsection{Gaussian Broadcast Channels and Worst-Case Noise\label{sec3.4}}
For a broadcast channel (BC), the common transmitter generates the channel input $\xvec$.
so its prescribed $R_{xx}$
need not have independent-user blocks.  Its receivers, however,
have separate locations.  Subsection~\ref{sec3.1}'s single-user
$R_{xx}\succeq0$ combines with a
worst-case-noise-autocorrelation matrix ($R_{nn} \rightarrow R_{wcn}$)) to create a single-user equivalent. 

When receivers cannot coordinate, receiver design may alter hypothetical
cross-receiver noise correlations while preserving every receiver's $R_{nn} (r)$ for $r=1,...,U$.
Such design accordingly need only preserve the block-diagonal set
\begin{equation}
 \mathcal R_{nn} \isdef \{R\succeq0:R_{rr}=R_{n_rn_r},\ r=1,\ldots,U\}.
 \label{eq:wcn-feasible}
\end{equation}
The off-diagonal terms can vary arbitrarily, but must preserve only positive-semi-definite conjugate symmetry.
For {\em any} prescribed input autocorrelation matrix $R_{xx}$, the
WCN autocorrelation optimization minimizes the linear log-mean-square-error-related (MMSE-GDFE) quantity 
\begin{align}
R_{wcn} \in\arg\min_{R\in\mathcal R_{nn}} \left\{ \;\;\;\;\;\;  \;\;\;\;\;\;  \;\;\;\;\;\; \right.\nonumber \\ \left.
\alpha \cdot \log_2 \Bigl|I+R^{-1/2}\cdot
   H\cdot R_{xx}\cdot 
   H^*\cdot R^{-*/2}\Bigr| \right\} \; .
 \label{eq:wcn-program}
\end{align}
The hypothetical cross-receiver correlations do not enter the
individual receivers' implementations.  At an interior nonsingular
optimum, the off-diagonal solution to (\ref{eq:wcn-program}) yields
\begin{equation}
 \mathcal S_{\rm wcn}=R_{\rm wcn}^{-1}
 -\left(H\cdot R_{xx}\cdot H^*+R_{\rm wcn}\right)^{-1},
 \label{eq:wcn-stationarity}
\end{equation}
where $\mathcal S_{\rm wcn}$ is block diagonal in the physical
receiver partition\footnote{Equation (\ref{eq:wcn-stationarity}) is a special case of a more general solution that accounts for the possibility of
singular WCN.   That singularity is not possible for the IC case of block-diagonal inputs, but it can occur for the BC. This singular situation introduces some notational complexity but does not change the conclusion that the WCN diagonalizes without loss the relevant receiver dimensions. See \cite{cioffi_ee379} or \cite{cioffi2023singular}.}.  Importantly, $R_{\rm wcn}$ need not itself be
block diagonal.  Boundary and singular solutions require the corresponding
positive-semidefinite optimality conditions and restriction to the
information-bearing support; see Appendix~\ref{app:wcn-local}
and Chapter~2, Appendix~2.C of~\cite{cioffi_ee379}. 

The overall all-user input autocorrelation matrix factors as $A\cdot A^*=R_{xx}$. 
The linear-MMSE backward-channel autocorrelation is
\begin{equation}
 R_b=\left(I+A^*\cdot H^*\cdot R_{\rm wcn}^{-1}
 \cdot H\cdot A\right)^{-1}
 =G^{-1}\cdot S_0^{-1}\cdot G^{-*}.
 \label{eq:wcn-backward}
\end{equation}
The matrix $G$ is always monic triangular and $S_0 \succeq I$ is diagonal. 
Importantly, the MMSE-GDFE developments consistently show there is no mutual-information loss by the linear
processing in what is Figure \ref{fig:gdfeintro}'s forward matrix filter ${\cal W}$.  The filter's output is a sufficient
statistic to decode the input $\xvec$ or equivalently $\vvec$.  This does not presume Gaussian optimality (yet to be shown shortly) and
simply represents a data-processing theorem \cite{coverbook} or reversibility theorem \cite{cioffi_ee379} lossless transformation for subsequent decoding. 

Appendix 2.C of \cite{cioffi_ee379} provides the
triangularizations of the WCN equivalent channel and also 
the input autocorrelation factor $A$. This setting formulates
the block-diagonal MMSE-GDFE receiver in terms of a rotation (unitary matrix) $Q_{\rm wcn}$ that derives from QR factorization of the WCN-whitened channel. 
and a corresponding block-diagonal factor $D_A$ to yield
\begin{equation}
 Q_{wcn}\cdot R_{wcn}^{-1}\cdot H\cdot A=D_A\cdot G,
 \label{eq:wcn-triangular}
\end{equation}
and also finds the block-diagonal feedforward matrix
\begin{equation}
 \mathcal W_{wcn} =S_0^{-1}\cdot G^{-*}\cdot A^*\cdot H^*
 \cdot R_{wcn}^{-1}=S_0^{-1}\cdot D_A^*\cdot Q_{wcn}.
 \label{eq:wcn-local-filter}
\end{equation}
${\cal W}_{wcn}$'s nonzero row elements are separately implementable at
corresponding receivers.  Those receiver operations also are
the aggregate and complete action of a larger WCN-single-user feedforward section.
Thus, since a Gaussian' channel's best single-user distribution for any $R_{xx}$ is Gaussian, and there is
now a single-user equivalent, the WCN-based BC's best input distribution is also Gaussian.
Its mutual information is
\begin{equation}
 {\cal I} (R_{xx},R_{\rm wcn})=\alpha \cdot \log_2 \left| S_0 \right| \; .
 \label{eq:wcn-information}
\end{equation}

The WCN construction is an analytical single-user
parallelization with the same prescribed input autocorrelation.
In a BC, coordinated transmission can directly implement
the corresponding triangular input factor $A$ and optionally nonlinear
precoding\footnote{$G_{unb} \isdef I + (S_0 -I)^{-1} \cdot S_0 \cdot (G-I)$} of $G_{unb}^{-1}$; the receivers still require only their respective
{\it local} feedforward operations.  Previously well-known Gaussian BC extremality \cite{weingarten2006capacity,mohseni2006converse} thus also follows
from the single-user Gaussian extremality together with the
applicable WCN/BC realization.  The singular and active-support
qualifications are stated in the following theorem, proved in
Appendix~\ref{app:wcn-local}.

\begin{theorem}[WCN stationarity and receiver-local diagonalization]
\label{thm:wcn-local}
GIVEN: A Gaussian vector channel has $\yvec=H\cdot\xvec+\nvec$
with prescribed $R_{xx}\succeq0$. The observation and
noise vectors partition into receiver blocks, with fixed positive-definite marginal noise
autocorrelation matrices $R_{n_rn_r}$.
The Gaussian single-user mutual information over the
positive-semidefinite noise autocorrelation matrices with these
prescribed diagonal blocks has a nonsingular optimum $R_{\rm wcn}$,
\begin{equation}
\label{eq:f1-stationarity}
R_{\rm wcn}^{-1}
 -(H\cdot R_{xx}\cdot H^*+R_{\rm wcn})^{-1}
 =S_{\rm wcn},
\end{equation}
where $S_{\rm wcn}$ is receiver-block diagonal. At a singular
optimum, the corresponding KKT condition includes a
positive-semidefinite multiplier for the constraint
$R_{\rm wcn}\succeq0$; see Appendix~\ref{app:wcn-local}.

TO PROVE: On the active primary-user subspace with WCN problem,
there exists a
factor $A$ of the active input autocorrelation and a corresponding monic
triangular $G$ such that the canonical stacked MMSE-GDFE
feedforward matrix ${\cal W}_{wcn}$ is receiver-block diagonal. Its receiver blocks
therefore use only their respective local observations. The
associated stacked Gaussian mutual information is
$\alpha\cdot \log_2 | S_0 |$, where $S_0$ is the diagonal innovation
matrix of the equivalent WCN channel.
\end{theorem}
The proof is in Appendix \ref{app:wcn-local}.  BC Gaussian optimality follows from MAC duality
\cite{cioffi_ee379,Vishwanath2003MIMODuality}, but the WCN-equivalent single-user channel here helps
prepare for the IC Gaussian optimality, which was previously not established and follows next.

\subsection{Gaussian Interference Channels with Joint Detection
\label{sec3.5}}
An interference channel (IC) imposes both the MAC and BC restrictions: the physical transmitters are independent,
requiring the block-diagonal $R_{xx}$ of
\eqref{eq:mmuc-blockdiag}. Also the IC receivers have separate locations,
requiring WCN analysis and receiver-local processing of (\ref{eq:wcn-local-filter}).  For \emph{every} admissible
block-diagonal set $\{R_{x_ux_u}\}$, the WCN construction creates a lossless equivalent single-user channel.  
This is because the ${\cal W}_{wcn}$ matrix combines a WCN-whitened matched-filter matrix operation (which matched matrix always provides
a sufficient statistic for detection) that is multiplied by an always-nonsingular following matrix (for nonzero noise). The always-invertible matrix
combines further no mutual-information loss with a block diagonal implementation.  The {\em combined} block-diagonal and WCN constraints reduce theIC's ${\cal I}_{\xvec ; \yvec}$ below
the same mutual information for either the MAC or BC alone. 

An implementation of the factor $A$ in~\eqref{eq:wcn-backward}--\eqref{eq:wcn-local-filter} mixes dimensions belonging to
different transmitters.  
However, the IC transmitters must implement a 
prescribed block-diagonal $R_{xx}$, not
$A$'s full matrix multiply.  The IC's physical inputs instead must use a
block-diagonal square root of the block-diagonal $R_{xx}$.  
Nonetheless, the (block-diagonal) input $\xvec$ has no mutual-information loss with respect to the alternative $A \cdot \vvec$ full-matrix implementation of $R_{xx}$.
The receiver may continue, as WCN-determined block-diagonal ${\cal W}$ to estimate $\vvec$ without loss,
and a lossless-invertible\footnote{The BC's $A^+$ need not be invertible because it has fat-matrix form $A = [ A_1 \; \hdots , A_U]$; however, the IC's block-diaognal $R_{xx}$ constraint forces $A$ to
be a square matrix that is invertible.} linear MMSE estimate $\hat{\xvec}$ may follow as
\[
\hat{\xvec} = A^{+} \cdot \hat{\vvec} \; .
\]
The $A^+$ is a pseudoinverse needed for the BC only when MIMO and the inputs $R_{xx} (u)$ are intentionally designed to be singular, $A^{-1}$ is sufficient for the IC.
However, $\yvec_r = H_r \cdot \xvec + \nvec_r$ for each receiver also has a lossless ${\cal W}_r$ for each partition ${\cal D}$ and ${\cal K}$ in (\ref{eq:local-lossless})'s detection of ${\cal M}_{\cal K}$.  The filter ${\cal W}_{wcn}$ preserves the information $\cI (\vvec \; \yvec )$, and thus also $\cI (\xvec ; \yvec )$ because of the 1-to-1 relationship.  An ML detector
decomposes without loss into one for each and all receivers $r$ for each and every ${\cal D}, {\cal K}$ partition.  The performance may be better for some partitions than others, but each is an optimum detector for the particular users' partition.   For any $R_{xx}$, $\cal D$, and $\cal K$, there is a maximum-likelihood (ML) detector that is optimum, although performance depends on the decoders' decisions on whom to detect and whom to treat as noise.  Good high-rate uses of this universally optimum ML detector require canonical design to find a good partition and breakpoint for $\cal D$ and $\cal K$.

For order $\pi_r$, this optimum joint detector may be more complex that a successive-decoding (GDFE) implementation.  However, for {\it arbitrary} block-diagonal $R_{xx}$, the joint decoding can be necessary (see related work in \cite{BandemerElGamalKim2015}) for largest weighted rate-sum achievement. 
Thus the MAC input restriction and the BC receiver-local
WCN construction provide a Gaussian single-user equivalent
for each prescribed block-diagonal $R_{xx}$ with joint decoding.  This establishes
the fixed-autocorrelation Gaussian input as optimum via application of Section \ref{sec3.1}'s single-user Gaussian result to this equivalent single-user channel.  
Again, this is for any $R_{xx}$; good rate choice and $R_{xx}$, $R_a$ choice remains, as in Section \ref{sec4} subject to the IC's block-diagonal restriction.

\begin{figure*}[t]
\begin{framed}
\vspace{-.2in}
\begin{center}
\includegraphics[width=0.9\linewidth]{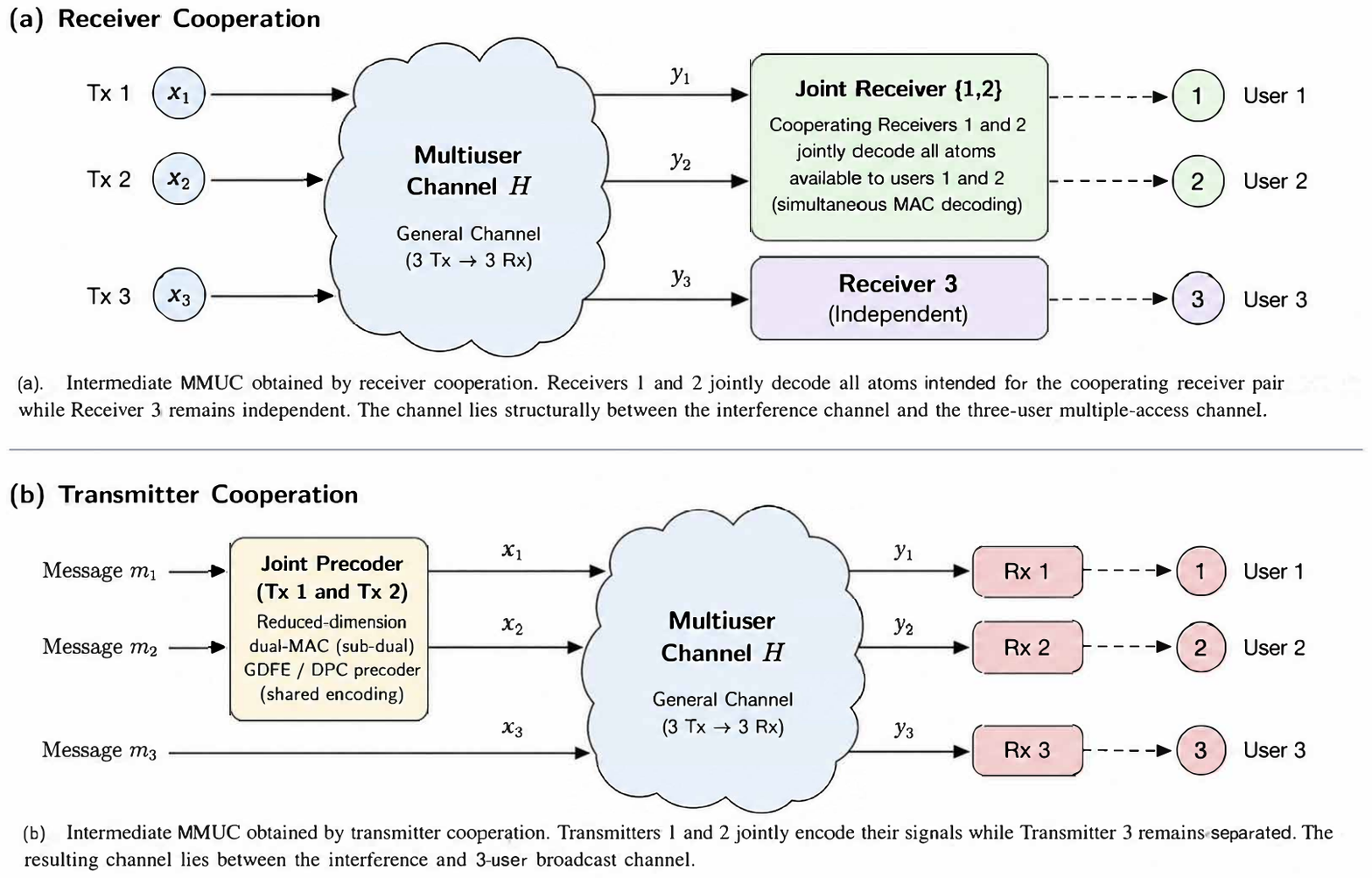}
\vspace{-.5in}
\caption{Transmitter and Receiver Subgroup Coordination.}
\label{fig:coop}
\end{center}
a). The single-user equivalent has a $3 \times 3$ diagonal input autocorrelation matrix with worst-case noise between receiver sets $S=\{1,2\}$ and $S'=\{3\}$.\\
b). The single-user equivalent has a $2 \times 2$ block-diagonal input autocorrelation matrix with worst-case noise across the three outputs.\\
\vspace{-.2in}
\end{framed}
\end{figure*}
Figure \ref{fig:coop} generalizes the IC generalizes to any MMUC.  The block-diagonal portion as per Figure \ref{fig:coop}(b) need apply only over those transmitter groups that need be separate; while Figure \ref{fig:coop}(b) illustrates that the worst-case-noise autocorrelation matrix need apply to the off-diagonal elements corresponding to receiver groups that cannot coordinate.  The same single-user equivalent follows, just with appropriate groupings at either or both of transmitter and receiver, thus establishing Gaussian optimality for all Gaussian-noise MMUCs. 

However, while the previous discussion establishes Gaussian optimality for all allowed $R_{xx}$, it does not yet maximize the Gaussian atomic energy distribution within each user's atomic rates and may require joint detection, which is more complex than GDFE-style successive decoding but often can be approximated. 
A final step optimizes the receiver-subset atomic
allocation and multicast atomic rate constraints and reintroduces the successive decoding over atoms, as developed next and further in Section \ref{sec4}.
The atoms are Gaussian without loss in the following achievability constructions (as the converse is already established through the single-user equivalence). 

\paragraph{Gaussian atomic-rate--energy optimization.}

A fixed nonnegative priority vector
$\boldsymbol\theta=(\theta_1,\ldots,\theta_U)$ defines a weighted rate sum.
For each user $u$, the atomic shares
$\{\lambda_a:a\in\mathcal A_u\}$ distribute that
user's {\it energy} among its private, pairwise, and
larger-multicast receiver-subset atoms. These shares are
distinct from synchronized time-sharing
probabilities, which combine complete
$U$-dimensional user-rate vectors. Each Gaussian atom has
autocorrelation $R_a\succeq0$ and rate $b_a$.
A multicast atom's rate must satisfy every receiver's
decoding constraints in
its designated subset.
\begin{equation}
 \sum_{a\in\mathcal T} b_a(q)
 \le f_{r,q}(\mathcal T;\{R_a(q)\}),
 \quad \mathcal T\subseteq\mathcal A_r,
 \label{eq:gauss-atomic-mac}
\end{equation}
where, for the Gaussian atomic construction,
\begin{equation}
 f_{r,q}(\mathcal T;\{R_a\}) \isdef
 \mathcal I\!\left(\mathbf x_{\mathcal T};\mathbf y_r
 \mid\mathbf x_{\mathcal A_r\setminus\mathcal T},Q=q\right).
 \label{eq:gauss-atomic-rank}
\end{equation}
(\ref{eq:gauss-atomic-mac}) and (\ref{eq:gauss-atomic-rank}) presume joint decoding of $\cal T$.
The effective Gaussian interference includes undecoded atoms.
An order-based successive realization replaces the subset
constraints by the appropriate receiver chain-rule increments;
then it bounds the common rate of atom $a$ by (\ref{eq:mmuc-Imin})'s minimum,
\emph{after} schedule averaging, of the receiver-specific increments.
The JD formulation above does not assume in advance
that one successive order achieves its optimum.

The weighted-rate-sum optimization follows the criteria below:
Each user $u$'s energy is ${\cal E}_u$,
and $\lambda_a$ denotes atom $(a\in\mathcal A_u)$'s energy fraction.
The Gaussian
atomic weighted-rate-sum optimization is then
\begin{equation}
\begin{aligned}
h_{\rm G}(\boldsymbol{\theta})
&=\max_{\{\lambda_a,\cdot R_a,b_a\}}
 \sum_{u=1}^{U}\theta_u \cdot 
 \sum_{a\in\mathcal A_u}b_a\\
\text{subject to}\quad
&\text{(\ref{eq:gauss-atomic-mac}), for every }r,\mathcal T,\\
&R_a\succeq0,\qquad b_a\geq0,\\
&\operatorname{tr}R_a=\lambda_a \cdot E_u,
 \quad a\in\mathcal A_u,\\
&\lambda_a\geq0,\quad
 \sum_{a\in\mathcal A_u}\lambda_a=1,\\
&0\leq {\cal E}_u\leq {\cal E}_{u,\max},\\
&\sum_{a\in\mathcal A_u}R_a
 =R_{x_ux_u},\\
&\hfill u=1,\ldots,U.
\end{aligned}
\label{eq:gaussian-atomic-opt}
\end{equation}
For a prescribed per-user autocorrelation, (\ref{eq:gaussian-atomic-opt})'s optimization may additionally
include a memoryless atomic-energy constraint (as is the case in all Section \ref{sec4}) results) $\sum_{a\in\mathcal A_u}R_a(q)=R_{x_ux_u}(q)$.
This satisfies (\ref{eq:mmuc-factorization}).
All (\ref{eq:gaussian-atomic-opt})'s energy and autocorrelation constraints are convex in the
atomic matrices.  For fixed $\{R_a(q)\}$, the receiver subset
rate inequalities define convex polymatroid regions, whose
intersection is also convex.  These observations alone do
\emph{not} prove joint convexity of
\eqref{eq:gaussian-atomic-opt} in the {\em rate} and {\em atomic-energy} variables.
Fixing $\boldsymbol\theta$ imposes, apart from time-sharing with schedule $Q$, that
every receiver's active decoding order for multicast atom's minimum mutual information over receivers.
For time sharing, there is a face that corresponds to at least two vertices with corresponding $\thetavec$ prioriities.
The points on the inferred hyperplane (line for just two vertices) follow this example
\[
\theta_{i} \cdot \bvec_i + \theta_j \cdot \bvec_j \rightarrow \theta_j \cdot \underbrace{\left[ \;\; \bvec_{ij} \;\; \right]}_{\frac{\theta_i}{\theta_j} \cdot \bvec_i + \bvec_j}  \; ,
\]
effectively introducing one joint rate vector $\bvec_{ij}$ for optimization that can be reversed and  split easily into a time-sharing of the two vertices $\bvec_i$ and $\bvec_j$.
The optimization then proceeds with this common macro user $ij$ as a single user and reduced number of users $U-1$.  This macro user may still have both atoms' decomposition as usual.
For this particular $\thetavec$ choice, the two (or more) users on this face (line in 2D) combine.  In general, the rate-region dimensionality reduces to $U-N_{vertices}+1$ where $N_{vertices}$ is minimum number of vertices necessary to combine to determine any facial point's averaged rate vector. 

\paragraph{Independent tones and cyclic prefixes.}
When an elementary MMUC vector symbol is an $N$-tone multitone block (see \cite{cioffi_ee379}, Chapter 4),
a sufficiently long cyclic prefix makes a linear finite-memory
channel circular over the useful block; independent
noise-block and symbol assumptions must still hold.
For a multitone model, each receiver's Gaussian
rank function adds across independent tones,
\begin{equation}
 f_{r,q}(\mathcal T)=\sum_{n=0}^{N-1}
 f_{r,q,n}(\mathcal T;\{R_{a,n}(q)\}),
 \label{eq:gauss-tone-sum}
\end{equation}
while the user energy constraint is
\begin{equation}
 \sum_q \alpha_q\sum_{a\in\mathcal A_u}
 \sum_{n=0}^{N-1}\operatorname{tr}R_{a,n}(q)
 \le\mathcal E_u^{\rm block}.
 \label{eq:gauss-tone-energy}
\end{equation}
The rates in \eqref{eq:gauss-tone-sum} are bits per useful
$N$-tone block; bits per transmitted sample include the
cyclic-prefix efficiency $N/(N+L_{\rm cp}) \rightarrow 1$ as $N$ becomes large. This remains
a well-designed symbol and fits (\ref{eq:mmuc-factorization})'s memoryless-multiuser channel definition.  The receiver
calculates a multicast-atom minimum \emph{after} summing
over the tones (and averaging over $Q$). 

Figure \ref{fig:coop}'s general MMUC falls somewhere between the MAC/BC duals and the IC.
The IC's combined WCN and block-diagonal constraints accordingly allow blocks to combine 
for common-location transmitters, while the WCN accordingly applies to off-block-diagonal elements
of $R_{nn}$ that correspond to receivers that cannot coordinate.  This intermediate MMUC
follows the same IC Gaussian generality's applicability. 

\begin{theorem}[Gaussian MMUC atomic rate region]
\label{thm:gauss-mmuc-conditional}
Equation  \eqref{eq:mmuc-gaussian} 's Gaussian MMUC has
admissible per-user energy constraints. For each user $u$,
$\mathcal A_u$ contains its receiver-subset atoms and $\lambda_a\geq 0$, with
$\sum_{a\in\mathcal A_u}\lambda_a=1$, are its atomic energy
shares. Each atom has $R_a\succeq0$, with
$\operatorname{tr}R_a=\lambda_a \cdot {\cal E}_u$ and
$\sum_{a\in\mathcal A_u}R_a=R_{x_ux_u}$.

Any admissible allocation $\{R_a\}$ has the receiver-wise
joint-decoding atomic region
\begin{equation}
\begin{aligned}
{\cal B}_{\rm JD}(\{R_a\})
=\bigcap_{r=1}^{U}
\Bigl\{\bvec_{\mathcal A}\succeq0:\quad&\\
\sum_{a\in\mathcal T}b_a
\le f_r(\mathcal T;\{R_a\}),\quad&
\forall\,\mathcal T\subseteq\mathcal A_r
\Bigr\}.
\end{aligned}
\label{eq:gauss-theorem-jd}
\end{equation}
For receiver orders $\boldsymbol\pi=(\pi_1,\ldots,\pi_U)$,
$\Delta_r(a;\pi_r,\{R_a\})$ represents the Gaussian chain-rule
increment for atom $a$ at receiver $r$ 
\begin{equation}
\cI_{\min}(a;\boldsymbol\pi,\{R_a\})
\isdef \min_{r\in\mathcal S(a)}
\Delta_r(a;\pi_r,\{R_a\}).
\label{eq:gauss-theorem-imin}
\end{equation}
The corresponding successive-decoding atomic-rate region is
\begin{equation}
\mathcal B_{\rm SD}(\{R_a\},\boldsymbol\pi)
=\{\mathbf b_A\succeq0: b_a\leq \cI_{\min}(a;\boldsymbol\pi,\{R_a\})\}.
\label{eq:gauss-theorem-sd}
\end{equation}
This construction projects atomic rates to user rates by
\begin{equation}
b_u=\sum_{a\in\mathcal A_u}b_a,\qquad u=1,\ldots,U.
\label{eq:gauss-theorem-region}
\end{equation}
The Gaussian atomic joint-decoding weighted-rate-sum value is the optimum
$h_{\rm G}(\boldsymbol\theta)$ of
\eqref{eq:gaussian-atomic-opt}, for each
$\boldsymbol\theta\succeq0$. The optimized successive-decoding
construction has the same weighted-rate-sum value whenever its receiver
orders and atomic allocations attain (\ref{eq:gauss-theorem-support})'s joint-decoding optimum.
Under that joint-to-successive equality, and the Gaussian
optimality/converse for the specified MMUC, the closed convexified
projection of these atomic-rate regions' intersection is its {\bf capacity region}:
\begin{equation}
{\cal C}_{\rm G}=\overline{\operatorname{conv}}
\left(\bigcup_{\mbox{\tiny admissible } \{R_a\},\boldsymbol\pi}
\operatorname{proj}_{\mathbf b}
\mathcal B_{\rm SD}(\{R_a\},\boldsymbol\pi)\right),
\label{eq:gauss-theorem-conv}
\end{equation}
Here the union ranges over admissible Gaussian atomic allocations and receiver orders. The convexification $\overline{\operatorname{conv}}$ includes Section \ref{sec2}'s common synchronized schedule $Q$. For a stationary Gaussian atomic realization, the schedule combines successive-decoding rate points and receiver-specific orders at the same atomic autocorrelation matrices $\{R_a\}$; it does not combine different atomic-energy allocations. Thus the schedule changes decoding and rate contributions without changing the prescribed Gaussian input distributions.
Its support function is
\begin{equation}
h_{\rm G}(\boldsymbol\theta)
=\max_{\bvec \in \mathscr C_{\rm G}}
\thetavec^* \cdot \bvec \; .
\label{eq:gauss-theorem-support}
\end{equation}
\end{theorem}

\begin{proof}
An admissible Gaussian atomic allocation is $\{R_a\}$.
For each receiver $r$, (\ref{eq:gauss-atomic-rank})'s Gaussian mutual-information function
$f_r(\mathcal T;\{R_a\})$
is submodular\footnotemark  in the decoded atom subset $\mathcal T$.
For an order $\pi_r=(a_1,\ldots,a_{m_r})$ of
$\mathcal A_r$, 
$\mathcal P_{r,k} \isdef \{a_1,\ldots,a_{k-1}\}$, then
\begin{equation}
\begin{aligned}
&\Delta_r(a_k;\pi_r,\{R_a\})\\
&\quad \isdef \cI (\xvec_{a_k};\yvec_r
          \mid\xvec_{\mathcal P_{r,k}})\\
&\quad=f_r(\left[ \mathcal A_r\setminus\mathcal P_{r,k} \right] ;
          \{R_a\})\\
&\qquad-f_r\bigl( \left[ 
    \mathcal A_r\setminus
    (\mathcal P_{r,k}\cup\{a_k\}) \right] ;
    \{R_a\}\bigr).
\end{aligned}
\label{eq:gauss-chain-increment}
\end{equation}
For every admissible input autocorrelation matrix $R_{xx}$, the WCN construction also supplies a Gaussian MAC equivalent to the single-user channel. This equivalent MAC has a successive MMSE-GDFE realization of its chain-rule rate vertices, with user order determined by the priority vector  $\thetavec$. Time-sharing among these vertices obtains intermediate MAC rate points without changing $R_{xx}$. The receiver-local WCN transformations and breakpoint decoding sets then provide the corresponding physical receiver realizations, subject to their multicast-atom decoding constraints. Thus the existence of successive orders follows from the IC's  MAC-WCN equivalent; the remaining design task selects and synchronizes receiver-specific orders and atomic rates that realize the desired user-rate point.

The set $\mathcal A_r$ specifies the atoms assigned to receiver $r$ by their decoding labels. Atoms outside this set remain Gaussian interference. The breakpoint rule provides a canonical priority ordering for the receiver's successive-decoding operations, but does not require lower-priority atoms to have zero rate. At a fixed atomic allocation, a multicast atom may have excess decoding capability at one receiver relative to its limiting receiver. Joint decoding can exploit this headroom through active subset-rate constraints, redistributing rates among jointly decoded atoms beyond what a single successive order realizes. The corresponding equivalent receiver-MAC rate point can instead has representation by a convex combination of chain-rule vertices. These vertices may admit different compatible priority orderings, even though the supported JD point has one specified priority vector $\thetavec$.

Here (\ref{eq:gauss-chain-increment}), expressed as a difference of Gaussian
MAC rank functions, is the defining chain-rule identity;
previously decoded atoms are given and atoms not decoded by
receiver $r$ remain Gaussian interference. Summing the increments
along (\ref{eq:gauss-chain-increment})'s order telescopes to $f_r(\mathcal A_r;\{R_a\})$.
Thus the successive MMSE-GDFE realizes the corresponding
receiver-MAC vertex. An atom decoded at several receivers
must use one code rate, which limits its feasible rate to
the smallest applicable increment, producing
\eqref{eq:gauss-theorem-imin} and
\eqref{eq:gauss-theorem-sd}.

For the receiver-MAC intersection, the separate chain-rule representations must also be compatible with common multicast-atom rates. The WCN single-user equivalent supplies the canonical Gaussian MAC ordering, while synchronized time-sharing combines the corresponding receiver-specific successive-decoding realizations. Each schedule state uses the same atomic autocorrelation matrices; only the receiver orders and rate contributions vary. A common schedule that reproduces the joint-decoding user-rate vector therefore provides a time-shared successive MMSE-GDFE realization without changing the (Gaussian) input energy allocation. Section \ref{sec4.1} gives the corresponding fixed-allocation linear-programming construction and feasibility test.

For joint decoding, all receiver subset inequalities must hold
simultaneously, yielding the intersection
\eqref{eq:gauss-theorem-jd}. The atomic sharing constraints
in \eqref{eq:gaussian-atomic-opt} optimize the autocorrelation matrix
assigned to each receiver-subset atom, subject to the user
energy and aggregate-autocorrelation constraints. Rate summing by
\eqref{eq:gauss-theorem-region} then provides user rates.
The receiver-wise chain rule proves that every selected MAC
vertex is successively realizable, but does not alone establish
that an optimum of the intersection of receiver MAC regions is
simultaneously a vertex for suitable receiver orders and then that this vertex may consequently need joint detection.
When atomic-allocation's, and possible reordering's, optimization attains (\ref{eq:wcn-feasible})'s maximum,
the joint-decoding weighted-rate-sum value, joint detection and successive-decoding/GDFE
have equal weighted rate sums for that $\boldsymbol\theta$; equality for
all $\boldsymbol\theta\succeq0$ gives equality of their closed
convexified downward-closed rate regions. 

Subsections \ref{sec3.1}-\ref{sec3.5}'s  single-user Gaussian
extremality and receiver-local equivalences
reaffirm the Gaussian multiuser optimality;
identification with the full MMUC capacity region uses the
equivalent single-user converse. Section \ref{sec2}'s synchronized convexification is already
specified and requires no atomic-share
interpretation of time-sharing weights. {\bf QED.}
\end{proof}
\footnotetext{For independent Gaussian atoms, each receiver's MAC rank function is submodular (it has diminishing marginal mutual-information increments), so its rate constraints form a polymatroid whose vertices follow from the chain rule and successive decoding.}
In layman's terms, the proof relates that there is nothing more that joint detection can provide after (\ref{eq:wcn-feasible})'s optimization. 
When both joint and successive detection align, then the method is constructive; work remains to outline a universal design's choices, which uses time-sharing to equate SIC performance to the JD
performance.  Fortunately, the exhaustive search over all orders and joint detection will not be necessary, as Section \ref{sec4} will show next.
\vspace{-.1in}
\section{Design Process and Examples\label{sec4}}
This section first develops a minimum-total-energy design for prescribed user rates, then illustrates Section~\ref{sec3}'s Gaussian capacity-region result through BC, IC, and intermediate MUC examples, including constructive receiver implementations,
single-user-equivalent interpretations, and maximum-rate-sum
specializations. 

\vspace{-.1in}
\subsection{Gaussian Atomic Rate--Energy Design\label{sec4.1}}

Section~\ref{sec3}'s Gaussian construction admits two complementary
numerical questions: (1) Given a target user-rate vector $\bvec^o\succeq0$,
what is the least total transmit energy that supports it?  (2) Conversely,
for prescribed individual transmitter-energy budgets $\mathcal E_u$
and a nonnegative priority vector $\thetavec$, what achievable user-rate
vector maximizes $\thetavec^* \cdot \bvec$?  Both questions use the same
receiver-subset atom set and the same joint-decoding (JD)
constraints.  This subsection first separates the linear fixed-allocation rate problem
from the generally nonlinear Gaussian atomic-energy optimization.

Each atom $a=(u,S)$ has $R_a\succeq0$ as its Gaussian input
autocorrelation and $b_a\geq0$ as its rate.  Independent atoms belonging
to user $u$ sum to its input autocorrelation, giving
\begin{equation}\label{eq:emin-rxx}
\rxx=\operatorname{blockdiag}\left(
 \sum_{a\in\mathcal A_1}R_a,\ldots,
 \sum_{a\in\mathcal A_U}R_a\right).
\end{equation}
A vector $\bvec_A$ collects the atomic rates, and a user--atom
incidence matrix $B$ has $B_{ua}=1$ when atom $a$ belongs to user $u$,
and zero otherwise.  Thus $\bvec=B \cdot \bvec_A$.  A prescribed
atomic autocorrelation-matrix allocation $\mathcal R=\{R_a\}$, enumerates each receiver's
nonempty decoded-atom subsets.  Their
Gaussian joint-detection (JD) inequalities from (\ref{eq:emin-rxx}) and
Section~\ref{sec3} are compactly
\begin{equation}\label{eq:iv-jd-matrix}
 \bvec_A\succeq0,\qquad A \cdot \bvec_A\preceq\mathbf c(\mathcal R),
\end{equation}
where each row of $A$ selects a receiver-decoded atom subset and the
corresponding entry of $\mathbf c$ is its conditional Gaussian
mutual information ${\boldmath c} ( \mathcal R )$, treating atoms outside that receiver's decoding
set as noise.  This notation includes multicast atoms: the same
$b_a$ appears in every applicable receiver inequality.
For $|\mathcal A|$ atoms, $B$ is $U\times | \mathcal A |$. If receiver $r$ decodes $d_r$ atoms, its $2^{d_r}-1$ nonempty subsets contribute that many rows to $A$, which therefore has $| \mathcal A|$ columns and $\sum_r(2^{d_r}-1)$ rows.
For MIMO channels, the same LP construction replaces the scalar Gaussian chain-rule increments by the corresponding conditional log-determinant increments computed from the fixed atomic autocorrelation matrices and the receiver's MMSE-GDFE processing.

\emph{Fixed-allocation, rate-tightening, and priorities.}  At fixed
$\mathcal R$, the weighted-rate support is the linear program
\begin{equation}\label{eq:iv-support-lp}
\begin{aligned}
 h_{\mathcal R}(\thetavec)=
 &\max_{\bvec_A\succeq0}\quad
       \thetavec^* \cdot B \cdot \bvec_A\\
 &\text{subject to}\quad A\cdot \bvec_A\preceq\mathbf c(\mathcal R).
\end{aligned}
\end{equation}
Its dual is
\begin{equation}\label{eq:iv-support-dual}
\begin{aligned}
 &\min_{\boldsymbol\mu\succeq0}\quad
     \mathbf c(\mathcal R)^* \cdot \boldsymbol\mu\\
 &\text{subject to}\quad A^* \cdot  \boldsymbol\mu\succeq B^* \cdot \thetavec \;.
\end{aligned}
\end{equation}
Complementary slackness\footnote{In this context, complimentary slackness refers to the $\cI_{min}$ as not binding at one or more multicast atoms' receivers.}  identifies the receiver-subset inequalities
that support the chosen atomic-rate point.  A related feasibility LP
adds $b^o \isdef B \cdot \bvec_A$ to (\ref{eq:iv-jd-matrix}); it tests
whether the prescribed user rates can match the \emph{fixed}
Gaussian atomic-rate allocation.  These LPs do not, by themselves, optimize the
atomic autocorrelation matrices.

For prescribed individual energy budgets, the complementary design is
\begin{equation}\label{eq:iv-maxric}
\begin{aligned}
 h_G(\thetavec;\overline{\boldsymbol{\mathcal E}})
 ={}&\max_{\mathcal R,\bvec_A}\ \thetavec^* \cdot B \cdot \bvec_A\\
 \text{subject to }&\ (\ref{eq:iv-jd-matrix}),\quad R_a\succeq0,\\
 &\sum_{a\in\mathcal A_u}\operatorname{tr}(R_a)
 \leq\mathcal E_u,\quad u=1,\ldots,U.
\end{aligned}
\end{equation}
The numerical routine denoted \texttt{maxRICU} \cite{cioffi_ee379} searches the outer
allocation and solves (\ref{eq:iv-support-lp}) again at its selected
fixed powers to verify the reported JD support.  The search over
$\mathcal R$ is not asserted to be jointly convex.

\emph{Minimum-energy design and its multipliers.}  The inverse
problem, implemented numerically by \texttt{minPICU}, is
\begin{equation}\label{eq:emin-design}
\begin{aligned}
 {\cal E}_{\min}(\bvec^o)
 ={}&\min_{\mathcal R,\bvec_A}
 \sum_{u=1}^U\sum_{a\in\mathcal A_u}\operatorname{tr}(R_a)\\
 \text{subject to }&\ (\ref{eq:iv-jd-matrix}),\quad R_a\succeq0,\\
 &B \cdot \bvec_A\succeq\bvec^o.
\end{aligned}
\end{equation}
The rate-requirement multipliers $\thetavec\succeq0$ enter its
Lagrangian through
\begin{equation}\label{eq:emin-lagrangian}
 {\cal L}=\sum_{u,a\in\mathcal A_u}\operatorname{tr}(R_a)
 +\thetavec^* \cdot (\bvec^o-B\cdot \bvec_A)+\cdots,
\end{equation}
where the omitted terms enforce the JD and autocorrelation matrix constraints.
Under the usual regularity and differentiability conditions,
$\theta_u$ is the marginal minimum-energy cost of increasing
$b_u^o$.  The energy-budget multipliers in
(\ref{eq:iv-maxric}) have a different interpretation and must not
be identified with $\thetavec$ without the appropriate optimality
conditions.  A supporting vector can be nonunique at a corner or
face; a strictly interior target point need not have a unique nonzero
supporting normal vector $\thetavec$.

This supporting-plane construction's geometric interpretation illustrates, at the fixed-budget region $\mathcal C(\boldsymbol{\mathcal E})$'s exposed boundary point $\bvec'$, that
\begin{equation}\label{eq:iv-tangent}
 \thetavec^* \cdot (\bvec-\bvec')\leq0,
 \qquad \bvec\in\mathcal C(\boldsymbol{\mathcal E}) \; .
\end{equation}
A companion admissibility calculation can test $\bvec^o$ against
individual budgets, measure the available rate margin, and determine
supporting priorities when the point is on an exposed boundary.
In particular, a minimum-\emph{sum}-energy solution that exceeds one
user's budget does not prove infeasibility under all other energy
allocations satisfying the separate budgets.

\emph{Fixed-energy synchronized SIC realization.}
Analysis may test JD feasibility and successive-decoding realizability  separately.  For each user $u$, the atomic energy
fractions $\lambda_{a(u)}$ specify a fixed Gaussian
allocation $\mathcal R$, with
$\sum_{a(u)}\lambda_{a(u)}=1$.  These fractions
remain unchanged throughout decoding.  The index $q$ specifies a
synchronized decoding state for a time-sharing fraction
$\alpha_q\geq0$, where $\sum_q\alpha_q=1$.
In state $q$, receiver $r$ uses its local successive-interference cancellation (SIC), more precisely the MMSE-GDFE process\footnote{Designer care
must use the MMSE cancellation, which only corresponds to  ``cancel the interference'' in special cases and more generally trades some residual interference for noise reduction.},  with order
$\pi_{r,q}$.  
$\Delta_{r,a,q}(\mathcal R)$ denotes the corresponding
Gaussian chain-rule increment for atom $a$, with zero
increment for atoms that receiver $r$ does not decode.
At the fixed allocation $\mathcal R$, the SIC
feasibility LP seeks common atomic rates $\bvec_A$
and synchronized state fractions $\boldsymbol\alpha$
satisfying
\begin{equation}\label{eq:iv-sic-lp}
 \begin{aligned}
 B\cdot \bvec_A&=\bvec^o,\qquad \bvec_A\succeq0,\\
 b_a&\leq\sum_q\alpha_q \cdot 
       \Delta_{r,a,q}(\mathcal R),
       \quad a\in\mathcal A_r,\quad r=1,\ldots,U,\\
 \alpha_q&\geq0,\qquad \sum_q\alpha_q=1.
 \end{aligned}
\end{equation}
For numerical evaluation, (\ref{eq:iv-sic-lp})'s SIC constraints follow through enumeration of each receiver's possible decoding orders. $d_r \isdef |\mathcal A_r|$ is the number of atoms decoded at receiver $r$, and $L_r \isdef d_r!$ is the number of possible orders. For each order $j$, the Gaussian chain-rule increments form a column of a matrix $V_r$ of size $|\mathcal A|\times L_r$, with zero entries for atoms not decoded by receiver $r$. The columns depend on the fixed allocation $\mathcal R$ and do not change during this linear program.

The constraint vector $\boldsymbol\gamma_r\succeq0$ contains the $L_r$ fractions assigned to receiver $r$'s orders, with $\mathbf 1^* \cdot \boldsymbol\gamma_r=1$. The program minimizes a common nonnegative shortfall $\delta$ subject to
\[
B\cdot\bvec_A=\bvec^o,\qquad \bvec_A\succeq0,
\]
\[
b_a\leq (V_r \cdot \boldsymbol\gamma_r)_a+\delta,
\qquad a\in\mathcal A_r,\quad r=1,\ldots,U.
\]
This is a linear program with $|\mathcal A|+\sum_r L_r+1$ variables. The user-rate requirements remain fixed, but the atomic rates may redistribute over a user's atoms to reduce the shortfall. A zero minimum, with independently verified constraint residuals, confirms SIC feasibility for the fixed atomic energies.

The receiver-specific order fractions are then combined into synchronized states with fractions $\alpha_q$ by partitioning their cumulative intervals. Each receiver retains its prescribed mixture, while all receivers use the same state boundaries. Every state's atomic autocorrelation matrices remain unchanged; only the decoding orders and corresponding rate contributions change. 
This retains (\ref{eq:mmuc-factorization})'s  memoryless defintion.
Fixing every atomic autocorrelation matrix is stronger than fixing each user's aggregate input autocorrelation.
Every state uses the same atomic energy fractions
$\lambda_{a(u)}$ and the same Gaussian atomic
autocorrelation matrices.  Each atomic code spans the complete
schedule; consequently, designs use neither time-sharing among different energy
allocations nor separately coded state-dependent
multicast rates.

The priority vector $\boldsymbol\theta$ orders the users in the equivalent single-user channel, with larger $\theta_u$ receiving higher priority. At physical receiver $r$, the corresponding breakpoint associates with its own user $r$. The receiver must recover its intended atoms, but may also decode atoms belonging to other users when their receiver-subset labels include $r$. Atoms outside $\mathcal A_r$ remain Gaussian interference. Thus the breakpoint separates the decoding obligations of the physical receiver from the larger priority ordering of the equivalent single-user channel.

For a prescribed atomic allocation, the receiver's JD constraints determine which combinations of decoded atoms can carry the assigned rates. A successive MMSE-GDFE implementation instead chooses an order of these atoms, with the conditional chain-rule increments given by (37). When one order does not support the required rates at every receiver, the fixed-allocation LP of (46) tests combinations of receiver-specific orders. The breakpoint guides the choice of these orders, while the LP checks their simultaneous feasibility. The resulting synchronized schedule does not alter the atomic autocorrelation matrices.

\emph{Receiver breakpoints and optimality confirmation.}  
The priority vector $\boldsymbol\theta$ establishes one user ordering, common to all receivers, with $\theta_U>\theta_{U-1}>\cdots>\theta_1$ for distinct priorities. Receiver $r$ has its own user as the breakpoint: it may decode atoms belonging to users $r-1,\ldots,1$, in the common priority order, but does not decode atoms belonging to higher-priority users $U,U-1,\ldots,r$. Thus receiver $1$ decodes only its own user's atoms, while receiver $U$ may decode atoms belonging to every user. The actual decoded atoms are further restricted by their receiver-subset labels and by the zero-energy atoms in the optimized allocation.
The breakpoint ordering can therefore restrict the receiver-specific SIC orders considered in (46), while the LP adjusts the atomic rates and synchronized order fractions at the fixed allocation. 
Section~\ref{sec3}'s
WCN single-user equivalent provides the necessary
canonical priority structure, but is not a physically pooled
receiver.  Each actual receiver observes only its own output,
selects its decoded atom set, and treats the remaining atoms as
Gaussian interference.  Its local JD constraints or synchronized
successive-decoding orders therefore determine the realizable
breakpoint.  Agreement of the optimized receiver breakpoints with
the WCN-equivalent ordering is a proposed final optimality test,
\emph{distinct} from convergence of the numerical allocation search.
Conclusion of such agreement's global minimum energy applies the
fixed-order energy minimum and tests the
single-user-equivalent lower bound for every admissible allocation.  The examples that followreport directly
verified rates, energies, subset residuals, and SIC schedules,
without using numerical convergence alone as a proof of global
optimality.  However, these and every other example tested met the WCN single-user order test upon convergence.  The author was
unable to find any non-convergent channel.  If the converged solution meets the WCN-equivalent order test, it is a global optimum.  

For $N$ independent frequency tones with cyclic prefix length $\nu$,
$R_a$ specializes to tone energies ${\cal E}_{a,n}$ and each Gaussian subset
mutual information is summed over tones and divided by
$c_b(N+\nu)$ to express bits per transmitted sample.  The useful
block energy is $\sum_{a,n} {\cal E}_{a,n}$; the corresponding energy per
transmitted sample under the stated CP convention is this sum
divided by $N$.  This multitone vector-symbol formulation does not
change the stationary Gaussian atomic powers across synchronized
SIC states.

For every admissible block-diagonal $R_{xx}$, the corresponding WCN construction provides a Gaussian MAC equivalent whose chain-rule decoding order follows the user priorities. Consequently, the receiver breakpoint order derives from this equivalent MAC rather than receiver specific. When the resulting rate--energy solution also attains the minimum energy over the admissible MAC-equivalent input allocations, it is globally optimal for the prescribed user-rate vector. 

\vspace{-.2in}
\subsection{Full Rank BC Example \label{sec4.2}}

A full-rank broadcast channel may possess multiple $\thetavec$-dependent subregions. This example illustrates how the ordering mechanism applies, while the active Gaussian recursion changes with $\thetavec$.

A $2 \times 2$ BC has full rank, $R_{nn} = I$, and a two-dimensional channel matrix:
\[
H =\left[ \begin{array}{cc} 1 & 2 \\ 1 & .5 \end{array} \right] \; .
\]
Figure \ref{fig:BCnondegraded} also shows two boundary subregiona for
the two priority orderings and also their time-shared boundary.
\begin{figure*}[t]
\begin{framed}
\centering
\includegraphics[width=0.9\linewidth]{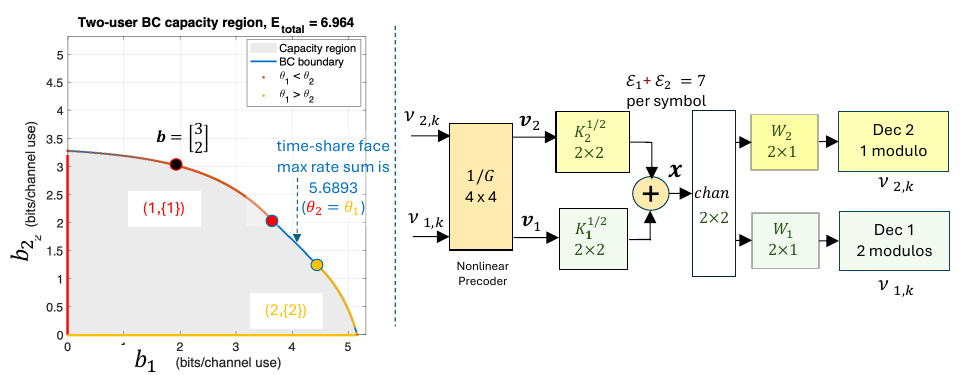}
\caption{Example $U=2$ Full-Rank Broadcast Channel Capacity Region.}
\label{fig:BCnondegraded}
\end{framed}
\vspace{-.1in}
\end{figure*}
This BC has a specific design for the rate vector $\mathbb b = [ 3 \;  2]$ with roughly total energy 7 units per transmitted symbol.  That rate point is on the region's boundary and uses those 7 units of energy.  Thus, Figure \ref{fig:BCnondegraded} also provides the capacity region for this 7 energy units.  The aggregate single-user (with WCN) equivalent MMSE-GDFE contains two local MMSE-GDFEs, one in upper 2 rows for user 2 and one in lower 2 rows for user 1.  Figure \ref{fig:BCnondegraded} shows the nonlinear precoder implementation at the transmitter, but that could also be implemented as the two separate GDFEs (upper rows at receiver 2 and lower rows at receiver 1).  For the BC, the WCN-dependent transmit matrix $A = [ R_2^{1/2} \; \; R_1^{1/2} ]$ is implementable at the transmitter (unlike the IC cases that occur later for a similar $2 \times 2$ channel). 
\[
G_{unbiased} = \left[ \begin{array}{cccc} 1 & -3,25 & 7.5759 & -3.7879 \\
0 & 1 & 2.331 & 1.1655 \\
0 & 0 & 1 & .5 \\
 0 & 0 & 0 & 1 \end{array} \right]  \; ,
\]
ard receiver gains (receiver 2 uses $2 \times 1$ upper two rows, while receiver 1 uses the lower two rows)
\[
{\cal W}_{unb} = \left[ \begin{array}{r} -1.9632  \\ 06041 \\ - - -  \\  0.4146 \\ .8292 \end{array} \right]  \; ,
\]

The maximum sum rate does not correspond to the $\mathbb b = [ 3 \; 2]$ because that point does not have equal $\thetavec$ elements. 
\vspace{-.2in}
\subsection{Rank-Deficient $2 \times 3$ BC with memory\label{sec4.3}} 

A second BC example with intersymbol interference, a $2 \times 3$ BC with $U=3$ users and two transmit antennas has the channel:
\begin{align}
H_{BC} (D) \;\;\;\;\;\;\;\; = \;\;\;\;\;\;\;\; &  \nonumber  \\
 \left[ \begin{array}{cc} 
.8D^3 & D^3 - D^2 \\
-.3D^2 +.2D & D^3 -D^2 -.63 D -.648 \\
D^3+.9D^2 & .5D^2 -.4D
\end{array}\right]  & \nonumber
\end{align}
and its (causal) dual
\begin{align}
H_{MAC-dual} (D) \;\;\;\;\;\;\;\; = \;\;\;\;\;\;\;\; &  \nonumber  \\
 \left[ \begin{array}{ccc} 
1+.9D & -.3D+.2D^2 & .8  \\
.5D-.4D^2 & 1-D-.63D^2 + .648D^3 & 1-D
\end{array}\right]  & \nonumber
\end{align}
where $D^k$ indicates $k$-sample delay. 
The three users respectively have lowpass, bandpass, and high-pass
characteristics.  
A 3-sample cyclic prefix allows the channel model to become memoryless between block symbols of 67 samples each with $N = 64$ tones.  The $64/67$ energy penalty was used in the design calculations.  As $N$ becomes large, this zeros loss, as does also any dimensionality loss.  This channel then exactly satisfies Equation (\ref{eq:mmuc-factorization})'s memoryless definition for the larger symbols.  If any of the $R_{xx} (u,n) \succeq 0$, the user is active - for the BC; unless a user has zero rate, all users are always active and each atom corresponds to one user (The BC only activates private atoms). 
Equal priority
$\thetavec=\frac13 \cdot [1\;1\;1]$ gives the maximum-rate-sum point.
Figure~\ref{fig:BC2x3ratesum} compares the canonical nonlinear design
with an all-linear MIMO design; the former nearly doubles the rate.
The low-subcarrier irregularity of the linear curve reflects its
nonconvex transmitter optimization.
\begin{figure*}[t]
\begin{framed}
\centering
\includegraphics[width=0.9\linewidth]{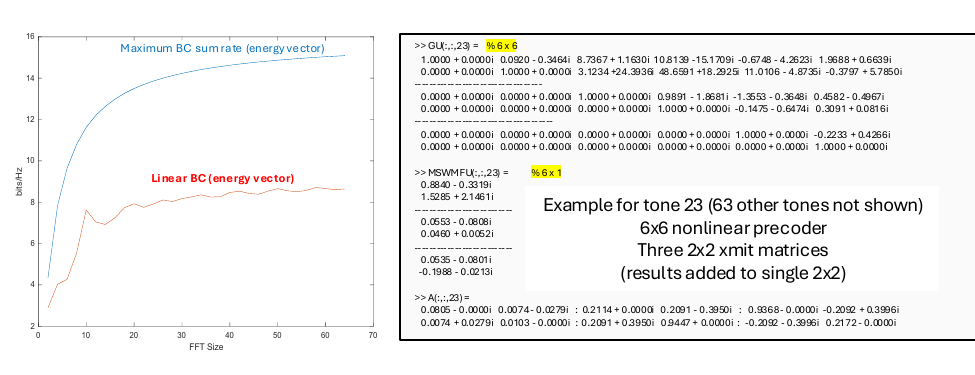}
\vspace{-.2in}
\caption{Example $2 \times 3$ Broadcast Channel rate-sum maximum point versus linear-only processing.}
\label{fig:BC2x3ratesum}
\end{framed}
\vspace{-.2in}
\end{figure*}
Each of the three users is two-dimensional, producing a $6\times6$
nonlinear precoder on each subcarrier.  Figure~\ref{fig:BC2x3ratesum}
shows the GDFE settings on representative subcarrier 23; the energies
and rates of the 64 subcarrier GDFEs add to the displayed totals.
The breakpoint ordering can therefore restrict the receiver-specific SIC orders considered in (46), while the LP adjusts the atomic rates and synchronized order fractions at the fixed allocation. The present numerical programs enumerate all receiver-specific atom orders, so their feasibility results do not confirm that the breakpoint-restricted orders alone are sufficient for convergence. A match to the WCN order confirms convergence.

\subsection{ISI-Free Gaussian Interference Channels\label{sec4.4}}

The following examples illustrate Gaussian atomic rate--energy IC design and
its fixed-allocation JD and successive MMSE-GDFE realizations.  The first
two channels have one scalar dimension per user and no intersymbol
interference (ISI).  
These are $2 \times 2$ and $3 \times 3$ respectively.
The final example uses an eight-tone vector symbol
to represent a frequency-selective channel.  The figure's plotted points are
numerical results of the Gaussian atomic optimization; feasibility at a
fixed allocation and global optimality of the nonconvex energy search
are distinct questions.

\noindent\textbf{Two-user ISI-free IC:}
Consider the real two-user channel
\[
 H=\begin{bmatrix}1&\sqrt{2}\\1/\sqrt{2}&1\end{bmatrix},
 \qquad {\cal E}_1={\cal E}_2=1.
\]
The atomic set comprises the private and two-receiver atoms of each
user.  Figure~\ref{fig:IC2stationary} shows 31 individually feasible
stationary-input JD rate samples.  Each sample optimizes the atomic
allocation for its prescribed boundary-rate condition and tests the
receiver MAC subset inequalities; the curve, consistent with (\ref{eq:mmuc-factorization}), does not
switch between different atomic energy allocations.
Figure~\ref{fig:IC2stationary}'s upper horizontal portion reaches $b_2=1$ over a range of $b_1$;
thereafter the boundary decreases toward the largest plotted $b_1$.
At all selected fixed allocations, Subsection~\ref{sec4.1}'s  SIC LP
confirms the global optimal
receiver-specific MMSE-GDFE order convergence, consistently synchronized with the same atomic
energies, attains the same JD user rates.  

\begin{figure}[!t]
\centering
\includegraphics[
    width=\columnwidth,
    trim=12 5 8 3,
    clip
]{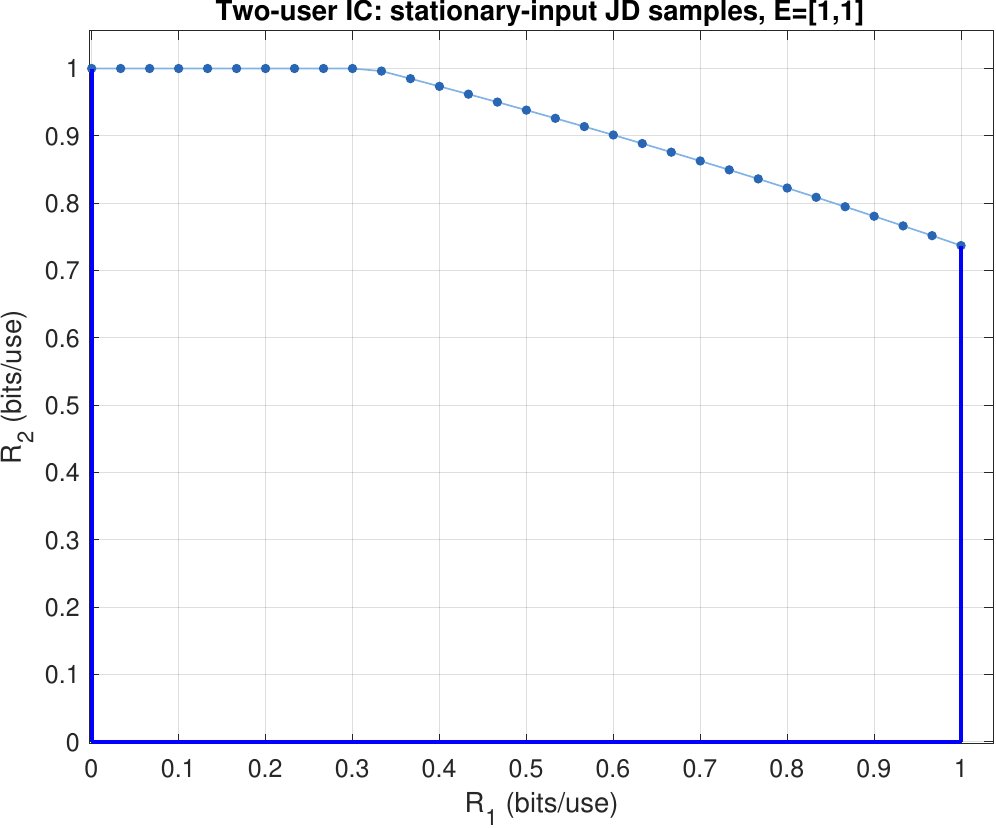}
\vspace{-.2in}
\caption{Numerically computed stationary-input JD rate samples
for a real two-user ISI-free Gaussian IC with individual
user energies ${\cal E}_1={\cal E}_2=1$.
Each sample uses a stationary Gaussian atomic allocation.}
\label{fig:IC2stationary}
\end{figure}

\noindent\textbf{Three-user ISI-free IC:}
For the real three-user channel designated Case F, let
\[
 H=\begin{bmatrix}3&7&2\\6&2&5\\4&3&5\end{bmatrix}.
\]
There are $U\cdot 2^{U-1}=12$ possible receiver-subset atoms, although only
some optimally energize at a particular rate point.  Figure~\ref{fig:3DCR}
shows the numerically computed three-user Gaussian JD rate-region
surface.  The marked point has user rates
$\bvec=[1\;3\;2]^*$.  A minimum-energy numerical JD search at this
point returned
\begin{align*}
(\mathcal E_1,\mathcal E_2,\mathcal E_3)&\simeq\ (14.724438,\,32.764313,\,142.581207),\\
\mathcal E_{\rm total}&\simeq190.069958.
\end{align*}
The receiver subset constraints were satisfied to numerical tolerance.
Table~\ref{tab:IC3atoms} lists the optimally energized atoms of that
allocation; values below the reporting threshold are omitted.  Here
$a=(u,S)$ identifies the transmitting user and the set of receivers
that decode the atom.
\vspace{-.1in}
\begin{table}[th]
\centering
\caption{Case F: fixed atomic energies and JD rates at $\bvec=[1\;3\;2]^t$.}
\label{tab:IC3atoms}
\footnotesize
\setlength{\tabcolsep}{3pt}
\begin{tabular}{c c r r}
\hline
Atom $a$ & $(u,S)$ & Energy & JD rate \\
\hline
4 & $(1,\{1,2,3\})$ & 14.724438 & 1.000000 \\
8 & $(2,\{1,2,3\})$ & 32.764313 & 3.000000 \\
9 & $(3,\{3\})$ & 0.043211 & 0.528386 \\
11 & $(3,\{2,3\})$ & $1.35\cdot10^{-7}$ & $1.17\cdot10^{-6}$ \\
12 & $(3,\{1,2,3\})$ & 142.537996 & 1.471613 \\
\hline
\end{tabular}
\end{table}
Atom 11 would likely not be implemented, but the tolerance is this use was $1e-7$, so its rate was above this. 
Table~\ref{tab:weightedRIC3}'s synchronized SIC realization was obtained
from a separate fixed-user-energy design.  It preserves
the user-energy vector
$({\cal E}_1,{\cal E}_2,{\cal E}_3)\simeq
(14.724438,32.764313,142.581207)$
and achieves the same target minimum-energy user-rate vector
$\bvec =[1\ 3\ 2]^t$ as the JD allocation in
Table~\ref{tab:IC3atoms}.  The two solutions differ slightly in the
distribution of user~3's energy among its atoms,
particularly atoms~9 and~11.  
Nevertheless, both attain the same user-rate
vector and individual user-energy budgets, with
synchronized SIC realizations.

Within each synchronized SIC realization, however,
the atomic energies are identical in every decoding
state.  Only the receiver-specific decoding orders and
their rates' time-sharing fractions change.  Consequently,
this construction does not time-share different Gaussian
input distributions and remains consistent with the
stationary-input MMUC model.


\begin{table}[t]
\centering
\caption{Three-user Gaussian IC: weighted-rate designs
with identical individual user-energy budgets.}
\label{tab:weightedRIC3}
\small
\begin{tabular}{lcc}
\toprule
 & Min-energy dual & Sum-rate \\
\midrule
$\theta_1$ & 0.337743 & 1 \\
$\theta_2$ & 0.343054 & 1 \\
$\theta_3$ & 0.319202 & 1 \\
\midrule
$R_1$ & 1.000000 & 0.074630 \\
$R_2$ & 3.000000 & 0.084223 \\
$R_3$ & 2.000000 & 5.841146 \\
\midrule
$\sum_u R_u$ & 6.000000 & 6.000000 \\
SIC states & 5 & 7 \\
JD/SIC feasible & Yes/Yes & Yes/Yes \\
\bottomrule
\end{tabular}
\end{table}

Table~\ref{tab:weightedRIC3} compares two weighted-rate
designs for the same Gaussian IC and the individual
user-energy budgets of Table~I.  The normalized
minimum-energy dual vector
$\boldsymbol{\theta}=(0.337743,0.343054,0.319202)$
reproduces the target rates $(1,3,2)$, with a
five-state synchronized SIC realization.  In contrast,
equal priorities $\boldsymbol{\theta}=(1,1,1)$,
corresponding to sum-rate maximization, produce rates
$(0.074630,0.084223,5.841146)$, with a seven-state
SIC realization.  Both designs achieve a sum rate of
$6$ with identical individual user-energy budgets,
although their atomic allocations differ.  Thus, even
a modest change in the priority vector can substantially
redistribute rates among users, with sum-rate optimization
strongly favoring the third user.

\begin{figure}[!t]
\centering
\includegraphics[width=0.92\columnwidth]{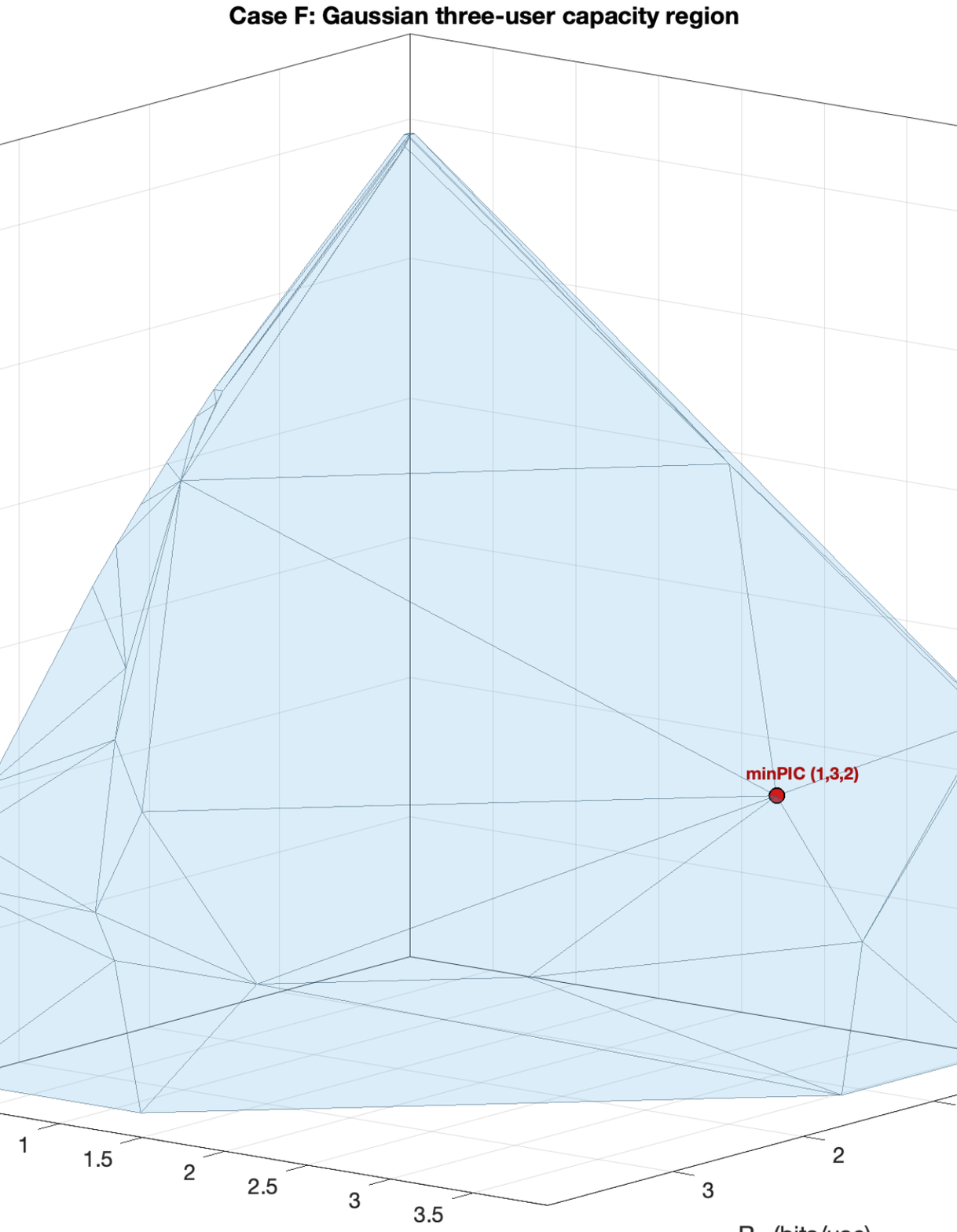}
\caption{Numerically computed Gaussian JD rate-region surface for the
three-user ISI-free Case F channel.  The marked point is the
$\bvec=[1\;3\;2]^t$ minimum-energy numerical design discussed in the
text.  The surface is a numerical boundary approximation, not a
certificate of global optimality at every plotted point.}
\label{fig:3DCR}
\end{figure}

\noindent\textbf{Frequency-selective two-user IC:}
The final example uses an eight-tone Gaussian IC.  An elementary
channel symbol is an $N=8$ tone block, with a cyclic prefix of one
sample.  Figure~\ref{fig:2x2ICmultitone} shows the numerically computed
rate boundary after optimization of the atomic energies across tones.
The near-horizontal upper portion and the curved descending portion
illustrate how frequency-dependent power allocations change the
region.  Each plotted operating point may have a different optimized
atomic allocation; this comparison across operating points is not
itself time-sharing among those allocations.  Any synchronized SIC
implementation of a chosen point instead retains its entire per-tone
atomic energy allocation in every decoding state.  Rates in the
figure are expressed per transmitted sample, including the indicated
cyclic-prefix convention.

\begin{figure}[!t]
\centering
\includegraphics[width=\columnwidth]{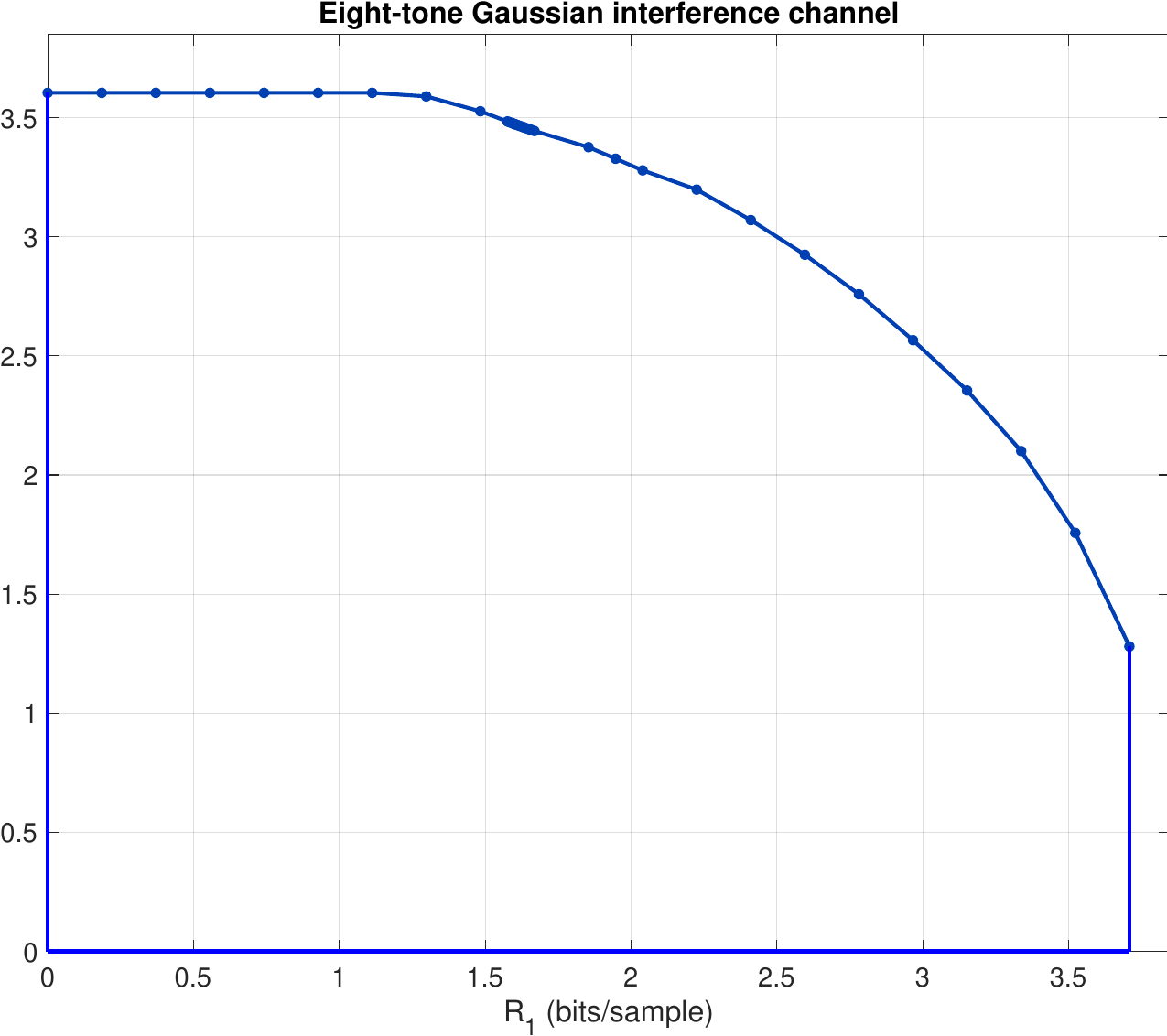}
\caption{Numerically computed eight-tone Gaussian IC rate boundary,
with one-sample cyclic prefix.  Atomic energies are optimized across
tones separately at each plotted operating point; synchronized SIC
states at a fixed point retain those same per-tone atomic energies.}
\label{fig:2x2ICmultitone}
\end{figure}
\vspace{-.1in}
\subsection{Multiuser Maximum Rate Sums\label{sec4.5}}

The maximum rate sum is a particularly simple capacity-region
boundary point.  It corresponds to equal user priorities, $\theta_1=\theta_2=\cdots=\theta_U$,
so that maximizing the supporting capacity-region tangent hyperplane reduces to
\begin{equation}
    b_{\rm sum}^{\max}
       = \max_{\bvec\in{\cal R}}
          \sum_{i=1}^{U} b_i .
\end{equation}
For a Gaussian MAC, Subsection~\ref{sec4.3}'s single-user equivalent
therefore reduces this optimization to simultaneous water-filling
subject to the autocorrelation matrix structures imposed by transmitter and
receiver coordination \cite{cheng1993multiuser}.  The iterative-waterfilling algorithm \cite{yu2004iterative} finds this optimum typically after a few cycles of waterfilling for each successive user, in any order, of a water-filling calculation where all other users in any water-filling substep are noise. 

With no remaining user-priority ordering, each transmitter water-fills
against the effective noise formed by the Gaussian noise plus all
other users' transmitted signals,
\begin{equation}
R_{{\rm in},i}
=
R_{\nvec\nvec}^{\rm wc}
+
\sum_{j\ne i}
H_j \cdot R_{\xvec_j\xvec_j} \cdot H_j^* .
\label{eq:muc-effective-noise}
\end{equation}
The corresponding MMSE receiver is equivalent, for this
water-filled solution, to treating all other users as Gaussian
interference.  Thus each user's autocorrelation matrix satisfies its
single-user water-fill with respect to $R_{{\rm in},i}$, while all
such water-fills are satisfied simultaneously.  This special
simplification occurs at the maximum-rate-sum point because the
equal $\thetavec$ elements remove the otherwise necessary
priority-dependent ordering.

With one coordinated transmitter, as for a BC, there is a single
input autocorrelation matrix $R_{\xvec\xvec}$ to optimize.  The maximum rate sum
consequently occurs with the usual matrix water-fill over the
single-user-equivalent Gaussian channel, simultaneously with the
worst-case-noise autocorrelation matrix associated with the uncoordinated
receivers \cite{yu2004iterative}.  Thus,
\begin{equation}
 \left(R_{\xvec\xvec}^{o},
       R_{\nvec\nvec}^{\rm wc,o}\right)
 =
 \arg\max_{R_{\xvec\xvec}}
 \;\min_{R_{\nvec\nvec}^{\rm wc}}
 I(\xvec;\yvec),
 \label{eq:bc-ratesum-wf}
\end{equation}
subject to the applicable transmit-energy constraint and the
worst-case-noise autocorrelation matrix constraints.  The autocorrelation matrix
$R_{\xvec\xvec}^{o}$ is the conventional single-user matrix
water-fill for the resulting worst-case-noise channel.  The
worst-case noise and water-fill are therefore found simultaneously, again
iterating a worst-case noise calculation with a water-filling step, as in \cite{yu2004sumcapacity} with more detail in \cite{cioffi_ee379}.
This corresponds to a saddle point between a convex descent and concave ascent procedure. 

When there are several uncoordinated transmitters, as in the IC,
the input autocorrelation matrix has again the block-diagonal structure
\begin{equation}
 R_{\xvec\xvec}
 =
 \operatorname{blockdiag}
 \left(
 R_{\xvec_1\xvec_1},\ldots,
 R_{\xvec_U\xvec_U}
 \right).
\end{equation}
The IC therefore combines simultaneous water-filling of its
block-diagonal transmit autocorrelation matrix with iteration of the corresponding
worst-case receiver-noise autocorrelation matrix.  

Consequently, the Gaussian MMUC maximum-rate-sum problem has the
particularly simple interpretation
\begin{equation}
\boxed{
\begin{aligned}
&\text{maximum Gaussian MMUC rate sum}\\[-1mm]
&\qquad\updownarrow\\[-1mm]
&\text{simultaneous water-filling with}\\
&\qquad R_{\xvec\xvec}^{\,\rm structured}
\text{ and } R_{\nvec\nvec}^{\,\rm wc}.
\end{aligned}}
\end{equation}
Thus BC, MAC, IC, and intermediate MUC rate sums differ only in the
transmitter- and receiver-cooperation constraints imposed on these two
autocorrelation matrices.

For the IC, the WCN-equivalent MAC retains the simultaneous water-filling solution, with separate water levels corresponding to the individual users' energy constraints and block-diagonal input autocorrelation. Alternating SWF and WCN updates may employ damping to avoid cycling, with convergence assessed by progressively narrowing the upper and lower bounds on the optimal rate sum.

\vspace{-.1in}
\section{Relation to Previous Gaussian IC Work\label{sec5}}

The Gaussian IC's receiver-subset atoms correspond to the private and common messages of the Han--Kobayashi construction \cite{han1981interference}. For two users, the four possible atoms are $(1,\{1\} )$, $(1, \{1,2\} )$, $(2,\{ 2\} )$, and $(2,\{1,2\} )$. The present formulation extends this representation to $U$ users with at most $U\cdot 2^{U-1}$ atoms, and optimizes their Gaussian autocorrelation matrices and rates subject to the individual receiver decoding constraints. The atomic representation is thus independent of a particular superposition-code construction. For this MMUC in (\ref{eq:mmuc-factorization}), this is the capacity region.  This work extends beyond the $2 \times 2$ IC, to any MMUC, including ISI and/or MIMO, and constructs a MMSE-GDFE-based SIC solution. 

Bandemer, El Gamal, and Kim \cite{BandemerElGamalKim2015} establish the importance of simultaneous nonunique decoding for interference-channel random-code ensembles. Joint decoding similarly provides the receiver-subset constraints used here to determine feasible Gaussian atomic rates. However, the subsequent fixed-autocorrelation SIC linear program tests finds that these same user rates also occur through time-shared receiver-specific successive MMSE-GDFE orders. Time-sharing can redistribute rates among a user's atoms without changing their Gaussian autocorrelation matrices. Through this process,  joint decoding and the resulting successive implementation achieve the same user-rate point with the same transmitted signaling statistics.  When the order matches, the solution is a global optimum. 

The resulting design distinction is important. Joint decoding provides a convenient description of the Gaussian receiver-rate constraints, while successive MMSE-GDFE offers an implementation using ordered decoding and interference cancellation. The numerical examples in Section IV illustrate their agreement at fixed atomic allocations. The Gaussian optimization and its receiver breakpoints further organize and simplify the search for suitable atomic allocations and decoding orders.
\section{Summary and Conclusions\label{sec6}}

A finite order-based characterization of the memoryless multiuser-channel
capacity region applies generally and simplifies canonical multiuser
system design. For linear Gaussian MUCs with trace-autocorrelation
constraints, the complete atomic representation establishes Gaussian
optimality, after which same-user atoms recombine into the conventional
Gaussian signaling architecture. The resulting single-user-equivalent
representation expresses the remaining channel distinctions through
structured input-autocorrelation and worst-case-noise matrices.

For an arbitrary $U$-user MUC, the complete representation contains
at most $U\cdot 2^{U-1}$ atoms, with reductions determined by channel
structure. In particular, the MAC and BC reduce to their conventional
$U$-user representations after Gaussianity is established.

The structural framework naturally encompasses MAC, BC, and IC building blocks for more general memoryless multiuser networks.   Relay channels, for example, could contain nested combinations of such
memoryless multiuser building blocks.  Their embedded structure enlarges the atomic representation exponentially, making the associated optimization correspondingly more complex, but a potential path forward.  Another implication of the atomic framework is that distributed agentic reporting of atom-activation possibilities/opportunities to higher protocol layers may enable further improvements in multiuser network design. 
 
The atomic framework applies to any block-wise memoryless channel, potentially reopening problems in network queuing, distributed computation, and overall system optimization within a unified structural framework. For Gaussian multiuser channels, where the characterization is fully constructive (Sections~\ref{sec3}--\ref{sec4}), it provides realizable communication architectures that achieve every rate tuple in the capacity region.
 
 {\bf Acknowledgment:}  The author would especially like to thank Prof. Merouane Debbah of Khalifa University and Prof Dongning Guo of Northwestern University for their encouragement and comments.  Similarly, I also want to thank Dr. Yuhan Zhou, Dr. Jung Hyun Bae, Dr. Jaber Borren, Dr. Hamid Saber, and Dr. Kee-Bong Song, all of Samsung Corporation for their significant comments and in particular their help on the clarity and detailed presentation.  Additionally, I would like to thank Prof Vince Poor of Princeton, Prof Amir Leshem of Bar-Ilan University,  Prof Ayfer Ozgur of Stanford, and Dr. Eric Ruzomberka of MIT Lincoln Laboratory for their thoughtful review and suggestions of this document.

\appendices
\section{WCN Stationarity and Receiver-Local MMSE-GDFE}
\label{app:wcn-local}
\begin{proof}[Proof of Theorem~\ref{thm:wcn-local}]
Notational simplification sets $K=H\cdot R_{xx}\cdot H^*$ and defines the feasible set
\begin{equation}
\mathcal R_{nn}=\{R\succeq0:\ [R]_{rr}=R_{n_rn_r}\}.
\end{equation}
For $R\succ0$ the objective, apart from the real/complex
normalization $\alpha$, is
\begin{equation}
F(R)=\log_2\det(K+R)-\log_2\det R.
\end{equation}
Its differential is
\begin{equation}
\mathrm d F=\frac{1}{\ln2}\operatorname{tr}\!\left[
\bigl((K+R)^{-1}-R^{-1}\bigr)\cdot\mathrm dR\right].
\end{equation}
The diagonal receiver blocks of $R$ are fixed. Thus every
admissible interior variation has zero diagonal blocks, and
stationarity forces the off-diagonal receiver blocks of
$R^{-1}-(K+R)^{-1}$ to vanish. Absorbing the logarithmic
normalization into the block-diagonal equality multiplier gives
\eqref{eq:f1-stationarity}.

For a boundary solution (singluar $R_{wcn}$) \cite{cioffi_ee379,cioffi2023singular}'s Chapter 2 KKT condition becomes
\begin{equation}
\label{eq:f1-singular}
R_{\rm wcn}^{+}-(K+R_{\rm wcn})^{+}
 +R_{\rm psd}=S_{\rm wcn},
\end{equation}
with $R_{\rm psd}\succeq0$ and $S_{\rm wcn}$ receiver-block
diagonal.  The extra term are the constraint multipliers to ensure positive semidefinite $R_{wcn}$.
Chapter 2, Appendix 2.C.1, restricts the factorization
to the active primary-user dimensions, for which the reduced
$R_{\rm wcn}$ is nonsingular. Secondary components may occupy
these same dimensions and are not thereby removed as messages.

On the active subspace let $A\cdot A^*=R_{xx}$ (with the
understanding that null-space components are discarded), and set
\begin{equation}
R_b=(I+A^*\cdot H^*\cdot R_{\rm wcn}^{-1}
 \cdot H\cdot A)^{-1}
 =G^{-1}\cdot S_0^{-1}\cdot G^{-*},
\end{equation}
where $G$ is monic triangular and $S_0$ is diagonal positive.
The matrix-inversion lemma transforms \eqref{eq:f1-stationarity} into
\begin{equation}
S_{\rm wcn}=R_{\rm wcn}^{-1}\cdot H\cdot A\cdot R_b
 \cdot A^*\cdot H^*\cdot R_{\rm wcn}^{-1}.
\end{equation}
Because $S_{\rm wcn}$ is receiver-block diagonal, choose a
receiver-block-diagonal unitary $Q_{\rm wcn}$ that diagonalizes it.
The two triangularization steps in Chapter 2, Appendix 2.C.1,
select the active factor $A$ and a diagonal $D_A$ satisfying
\begin{equation}
Q_{\rm wcn}\cdot R_{\rm wcn}^{-1}\cdot H\cdot A
 =D_A\cdot G.
\label{eq:f1-triangular}
\end{equation}
Consequently the stacked MMSE-GDFE feedforward matrix is
\begin{align}
\mathcal W
 &=S_0^{-1}\cdot G^{-*}\cdot A^*\cdot H^*
 \cdot R_{\rm wcn}^{-1}\\
 &=S_0^{-1}\cdot D_A^*\cdot Q_{\rm wcn},
\label{eq:f1-filter}
\end{align}
which is receiver-block diagonal because $S_0$, $D_A$, and
$Q_{\rm wcn}$ have that structure on the selected active
receiver coordinates. Finally,
\begin{equation}
\mathcal I_{\rm wcn}(\xvec;\yvec)
 =\alpha\log_2\det(R_b^{-1})
 =\alpha\log_2\det S_0,
\end{equation}
since $\det G=1$. This proves the claimed stacked equivalent
and local feedforward structure on the active subspace.

The receiver-block-diagonal feedforward matrix therefore
provides receiver-local sufficient statistics for the
corresponding Gaussian detection problems. The common
WCN-equivalent priority order determines the receiver-specific
decoding orders and breakpoints. Each receiver implements
the resulting order with its own MMSE-GDFE feedforward
and feedback settings, or equivalently uses joint decoding
where required. No exchange of observations among
receivers is necessary.
\end{proof}
\vspace{-.2in}

\bibliographystyle{ieeetr}
\bibliography{refs}

\end{document}